\documentclass{article}
\usepackage[margin=1.5in]{geometry}
\usepackage[utf8]{inputenc}
\usepackage{graphicx}
\usepackage{amsmath,amssymb,amsfonts,amsthm}
\usepackage{enumerate}
\usepackage{hyperref}

\newcommand{\email}[1]{\protect\href{mailto:#1}{#1}}

\def \FigPath {./figs/}

\usepackage[numbers,sort&compress]{natbib}
\newtheorem{theorem}{Theorem}[section]

\newcommand{\N}{\mathbb{N}}
\newcommand{\C}{\mathbb{C}}
\newcommand{\expon}[1]{\mathrm{e}^{#1}}
\newcommand{\imag}{\mathrm{i}}
\newcommand{\norm}[1]{\left\lVert#1\right\rVert}
\newcommand{\abs}[1]{\left\lvert#1\right\rvert}
\newcommand{\cardinality}[1]{\abs{#1}}

\newcommand{\spaceimdim}[1]{\C^{{#1}\times {#1}}}
\newcommand{\spaceim}{\spaceimdim{n}}
\newcommand{\spacefrdim}[1]{\C^{#1}}
\newcommand{\spacefr}{\spacefrdim{m}}
\newcommand{\fourier}{\mathcal{F}}
\newcommand{\mriop}{\mathcal{A}}
\newcommand{\pseudoinv}[1]{{#1}^\dagger}
\newcommand{\mriopinv}{\pseudoinv{\mriop}}

\newcommand{\recon}{\mathrm{rec}_{\mathrm{TV}}}
\newcommand{\reconeta}[1]{\recon{\left({#1}; \eta\right)}}
\newcommand{\tira}{\mathrm{Tira}}
\newcommand{\tirafl}{\mathrm{TiraFL}}
\newcommand{\mask}{\Omega}
\newcommand{\maskop}{\mathrm{P}_{\mask}}
\newcommand{\maskopn}{\mathrm{P}_{\mask_n}}
\newcommand{\maskopinv}{\mathrm{P}_{\mask}^\dagger}

\DeclareMathOperator*{\argmax}{arg\,max}
\DeclareMathOperator*{\argmin}{arg\,min}

\title{Localized adversarial artifacts for compressed sensing MRI}
\author{
Rima Alaifari\thanks{
                Seminar for Applied Mathematics, Department of Mathematics, ETH Z\"urich (\email{rima.alaifari@math.ethz.ch}, \email{tandri.gauksson@math.ethz.ch}).
                }
\and Giovanni S. Alberti\thanks{
                MaLGa Center, Department of Mathematics, University of Genoa (\email{giovanni.alberti@unige.it}).
                }
\and Tandri Gauksson\footnotemark[1]
}
\date{}

\begin{document}

\maketitle

\setcitestyle{authoryear,round,semicolon}

\begin{abstract}
As interest in deep neural networks (DNNs) for image reconstruction tasks grows, their reliability has been called into question
\citep{antun2020instabilities,gottschling2020troublesome}.
However, recent work has shown that, compared to total variation (TV) minimization, when appropriately regularized, DNNs  show similar robustness to adversarial noise in terms of \(\ell^2\)-reconstruction error
\citep{genzel2020solving}.
We consider a different notion of robustness,
using the \(\ell^\infty\)-norm,
and argue that localized reconstruction artifacts are a more relevant defect than the \(\ell^2\)-error.
We create adversarial perturbations to undersampled magnetic resonance imaging measurements (in the frequency domain) which induce severe localized artifacts in the TV-regularized reconstruction.
Notably, the same attack method is not as effective against DNN based reconstruction.
Finally, we show that this phenomenon is inherent to reconstruction methods for which exact recovery can be guaranteed, as with compressed sensing reconstructions with $\ell^1$- or TV-minimization.
\end{abstract}

\setcitestyle{numbers,square,comma}

\section{Introduction}

Following the success of deep learning in computer vision, deep neural networks (DNNs) have now found their way to a wide range of imaging inverse problems \cite{mccann2017review,AMO19,ongie2020deep}.
In some applications, learning the distribution of images from data is the only option.
In others, existing methods based on hand-crafted priors are well established.
Magnetic resonance imaging (MRI) reconstruction, for which sparsity-based methods have been highly successful, is an example of the latter \cite{lu2008csmri}.
However, recent work suggests that image quality can be improved and computation times shortened significantly by the use of DNNs in MRI reconstruction
\cite{chen2022ai}.

At the same time, it is well known that DNNs trained for image classification admit so-called adversarial examples---images that have been altered in minor but very specific ways to change the label predicted by the network \cite{biggio2013evasion, szegedy2013intriguing}.
In \cite{antun2020instabilities}, it was discovered that DNNs used in inverse problems (MRI and computed tomography) exhibit similar behavior.
Namely, the authors show that perturbing the measurements slightly can lead to undesirable artifacts in the image reconstructed by the network and that the same perturbations do not cause problems for state-of-the-art compressed sensing methods.
On the other hand, \cite{genzel2020solving} shows quantitatively that DNNs can be made robust, to a comparable level with total variation (TV) minimization, by injecting statistical noise to the measurement data during training.
Here, robustness is measured by the mean relative reconstruction error as a function of relative (adversarial or statistical) noise level, where both are defined by the \(\ell^2\)-norm.
While this training procedure (known as ``jittering'') increases stability,
it comes at the cost of accuracy, as the recovery of out of distribution features is compromised \cite{colbrook2022difficulty,genzel2020solving}.
This trade-off can be controlled by tuning the noise level used at training time
\cite{genzel2020solving}.

Although DNNs and TV-regularized reconstruction behave similarly to adversarial perturbations from this quantitative perspective, the reconstruction artifacts seen in \cite{genzel2020solving} are qualitatively very different.
TV-minimization suffers from global degradation of image quality (due to staircasing effects), while DNNs tend to ``hallucinate'' new meaningful features to the image.
The latter type of artifact is arguably worse, but it is only severe at relatively high noise levels.
Moreover, the introduced features are already visible in the adversarial noise and are therefore enhanced, but not created, by the network.

In the current work, we focus on the two-dimensional compressed sensing MRI problem \cite{lu2008csmri} and aim to create adversarial perturbations for TV-regularization that result in more localized reconstruction artifacts
as displayed in Figure \ref{fig:mask_structpert}.
To this end, we simply modify the attack method used in \cite{genzel2020solving} by replacing the loss function with a weighted seminorm, with a weight vector supported on a targeted location.
Several locations are targeted, and the one on which the largest artifact appears is selected.
Our experimental results show that such localized adversarial perturbations for TV-regularization do exist, and their \(\ell^\infty\)-norm is directly related to the subsampling factor of the MRI operator.
This attack method has a milder effect when applied to DNNs.

The resulting artifacts for TV-regularization are often manifested as isolated spikes in the image.
These spikes should be interpreted as toy artifacts (as opposed to the hallucinations of a DNN) but can be generalized to more realistic features (see the discussion at the end of Section \ref{sec:maths}).
We provide a mathematical justification for the appearance of such spikes, based on the theory of compressed sensing. Curiously enough, the same positive compressed sensing results on undersampled MRI guaranteeing exact recovery of sparse signals have a negative counterpart regarding the appearance of sparse artifacts. While the analysis we provide is in the context of MRI, the motivations behind this phenomenon are more general and are mainly due to the fact that the forward operator in undersampled MRI has a nontrivial kernel (see also \cite{gottschling2020troublesome}), which happens in many inverse imaging problems.

\begin{figure}
\centering
\includegraphics[scale=0.65]{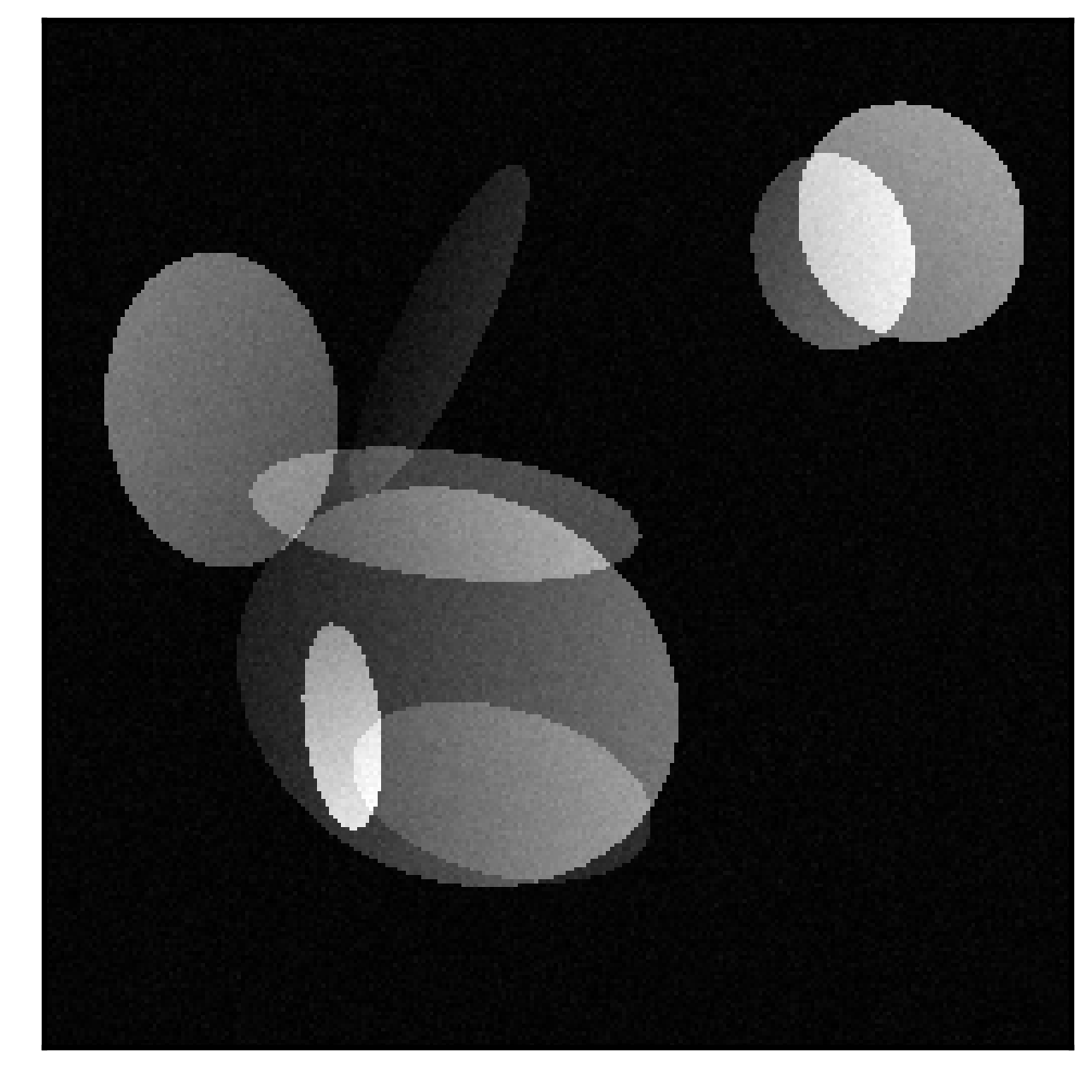}
\includegraphics[scale=0.65]{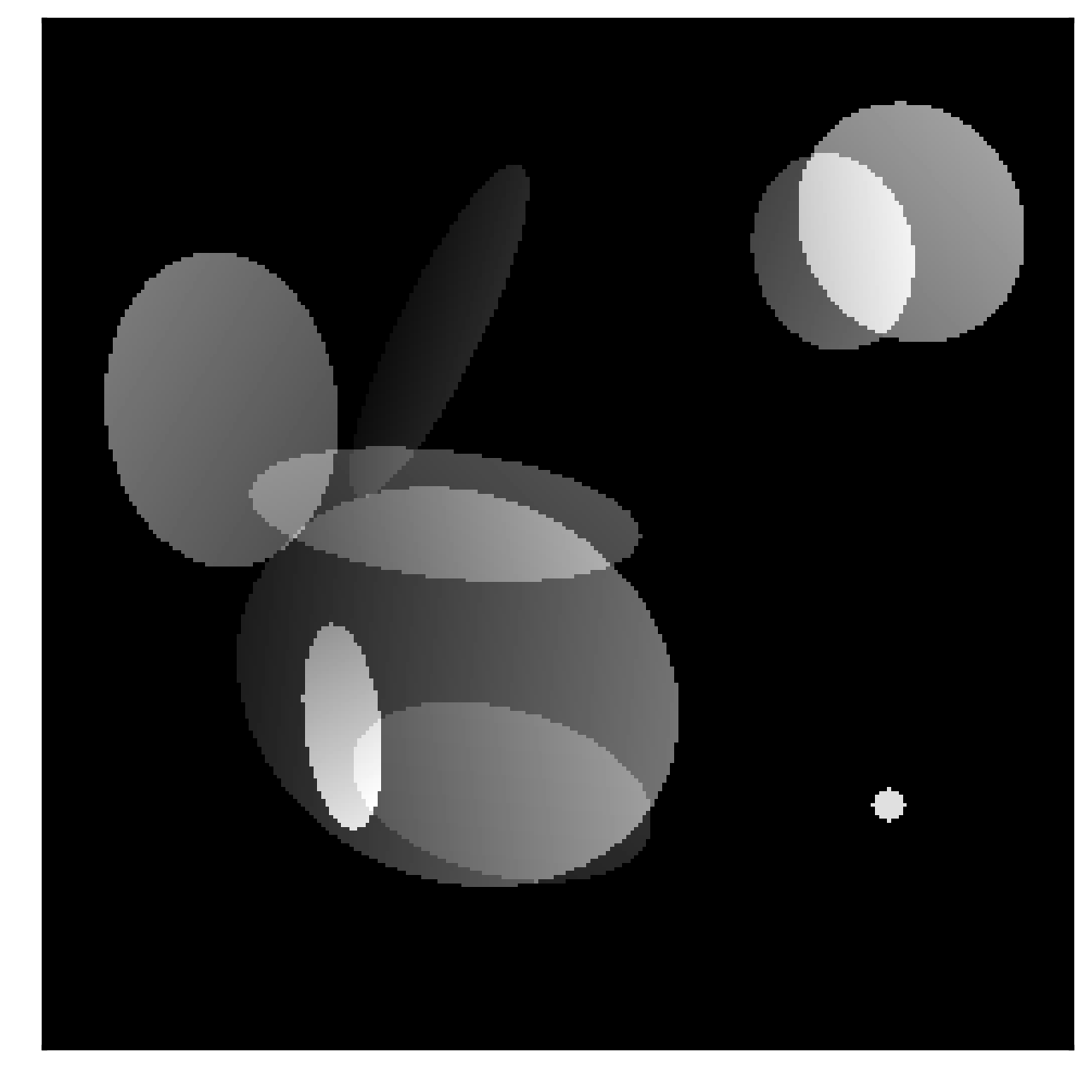}
\caption{
The same image corrupted by two different perturbations of equal \(\ell^2\)-norm (\(10\%\) of the norm of the original image): Gaussian noise (left) and a potentially meaningful feature (right).
The goal of this work is to find small adversarial perturbations to MRI measurements such that artifacts in the spirit of the latter appear in the reconstruction process.
}
\label{fig:mask_structpert}
\end{figure}

We note that in regularized reconstruction for imaging applications, the standard benchmark for stability is, and always has been, a quantitative one described by a norm estimate in some chosen (appropriate) norm. At the same time, it is interesting to see the influence of the adversarial attacks research on DNN-based image reconstruction. It seems a valid argument to point out that suitable attacks can lead to artifacts that are qualitatively relevant; after all, the objects of interest are images. So to us, this also opens the discussion about benchmarks for classical regularization algorithms.  One might argue that in that case, qualitative benchmarks (which are of course harder to define) should also play a relevant role.
It is natural to wonder whether it is possible to design a measure of the perturbations that could be both quantitatively computable and qualitatively meaningful.
Furthermore, future work could generalize the experimental results and the mathematical insights of this paper to different kinds of artifacts and to more general inverse problems.

The paper is structured as follows:
In Section~\ref{sec:csmri} we introduce the compressed sensing MRI model,
in Section~\ref{sec:adv} we formulate our adversarial attack method,
in Section~\ref{sec:results} we present the results of our numerical experiments,
and
in Section~\ref{sec:maths} we give a mathematical explanation for our observations.

The code for the experimental part of the paper builds on the code from \cite{genzel2020solving} and is available at \url{https://gitlab.math.ethz.ch/tandrig/localadvmri}.

\section{Compressed sensing MRI}
\label{sec:csmri}
The goal of MRI is to recover an object's density from its Fourier coefficients.
In the fully sampled case, this problem is readily solved by applying the inverse Fourier transform.
However, acquisition times can be reduced significantly by undersampling along non-Cartesian trajectories in the frequency domain.
This leads to an underdetermined linear system, and it is clear that some additional assumptions need to be made on the object density in order to get a good reconstruction
\cite{lu2008csmri}.

Let us model the cross section of an object by an image of resolution \(n\times n\) for some \(n\in\N\).
Let \(\fourier\colon\spaceim\to\spaceim\) be the two-dimensional discrete Fourier transform, and let \(\maskop\colon\spaceim\to\spacefr\) be a projection onto a set of \(m\geq 1\) indices, \(\mask\subseteq\{(i,j)\colon i,j=1,\ldots,n\}\).
The forward operator of the subsampled MRI problem is \(\mriop=\maskop\circ\fourier\),
and the measured values corresponding to an image \(x\in\spaceim\) are
\begin{equation}
    \label{eq:noisymeasurements}
	y = \mriop x + \epsilon,
\end{equation}
where \(\epsilon\) is zero-mean random noise which we assume is bounded in \(\ell^2\)-norm by some constant \(\eta>0\).
In order to recover \(x\) from the measurement vector, we must search for an image \(z\) with \(\norm{\mriop z - y}_2 \leq \eta\).
Since \(m<n^2\), the equation \(y=\mriop z\) does not have a unique solution.
The pseudoinverse \(\mriopinv = \fourier^{-1}\circ\maskopinv\) certainly provides a solution, but the resulting image may be of low quality or exhibit severe aliasing artifacts (depending on \(\mask\)) and not at all represent the desired object's density (see Figure \ref{fig:mask_radial40}).
Instead, we must impose some conditions on \(z\) based on our a priori knowledge of \(x\).

\begin{figure}
\centering
\includegraphics[scale=0.65]{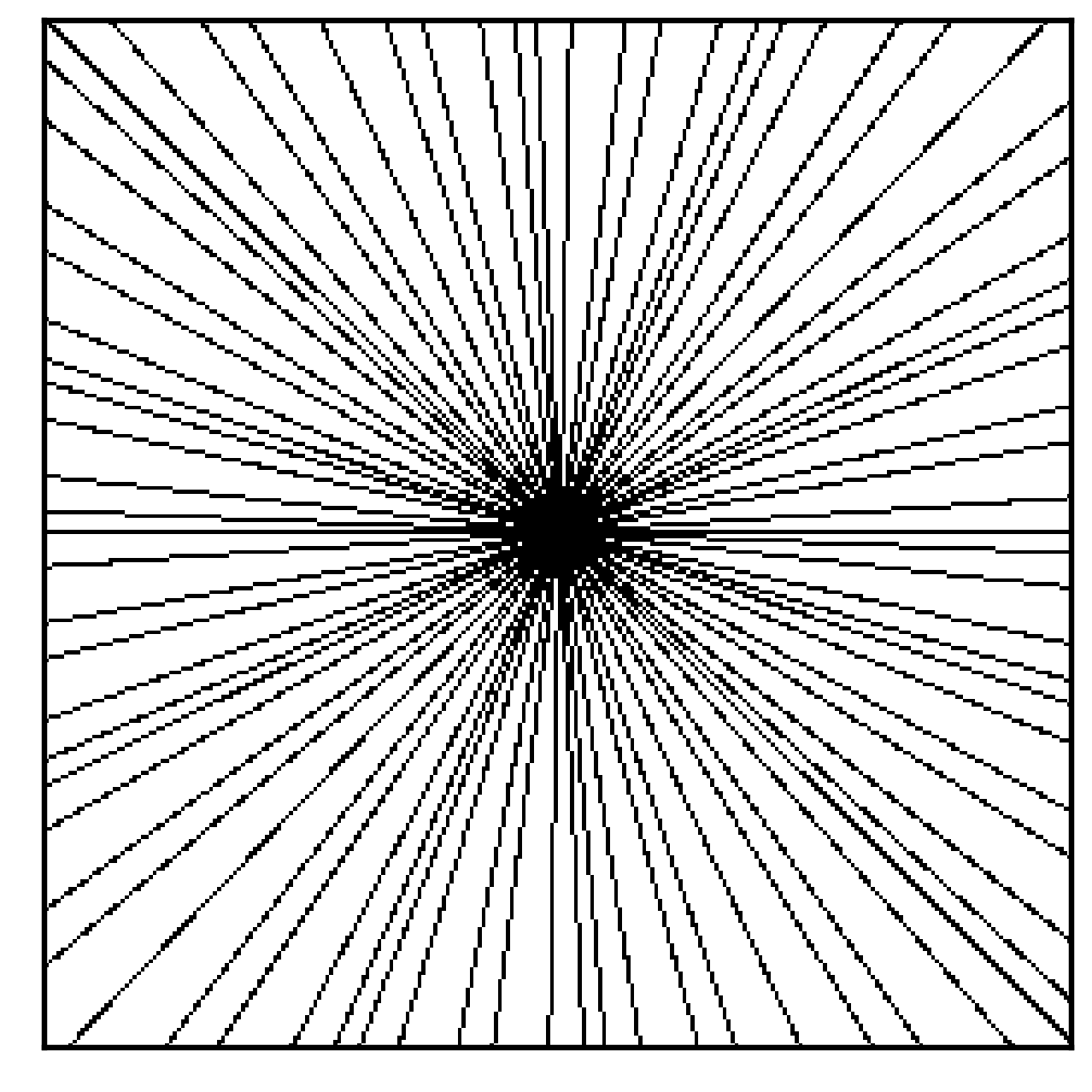}
\includegraphics[scale=0.65]{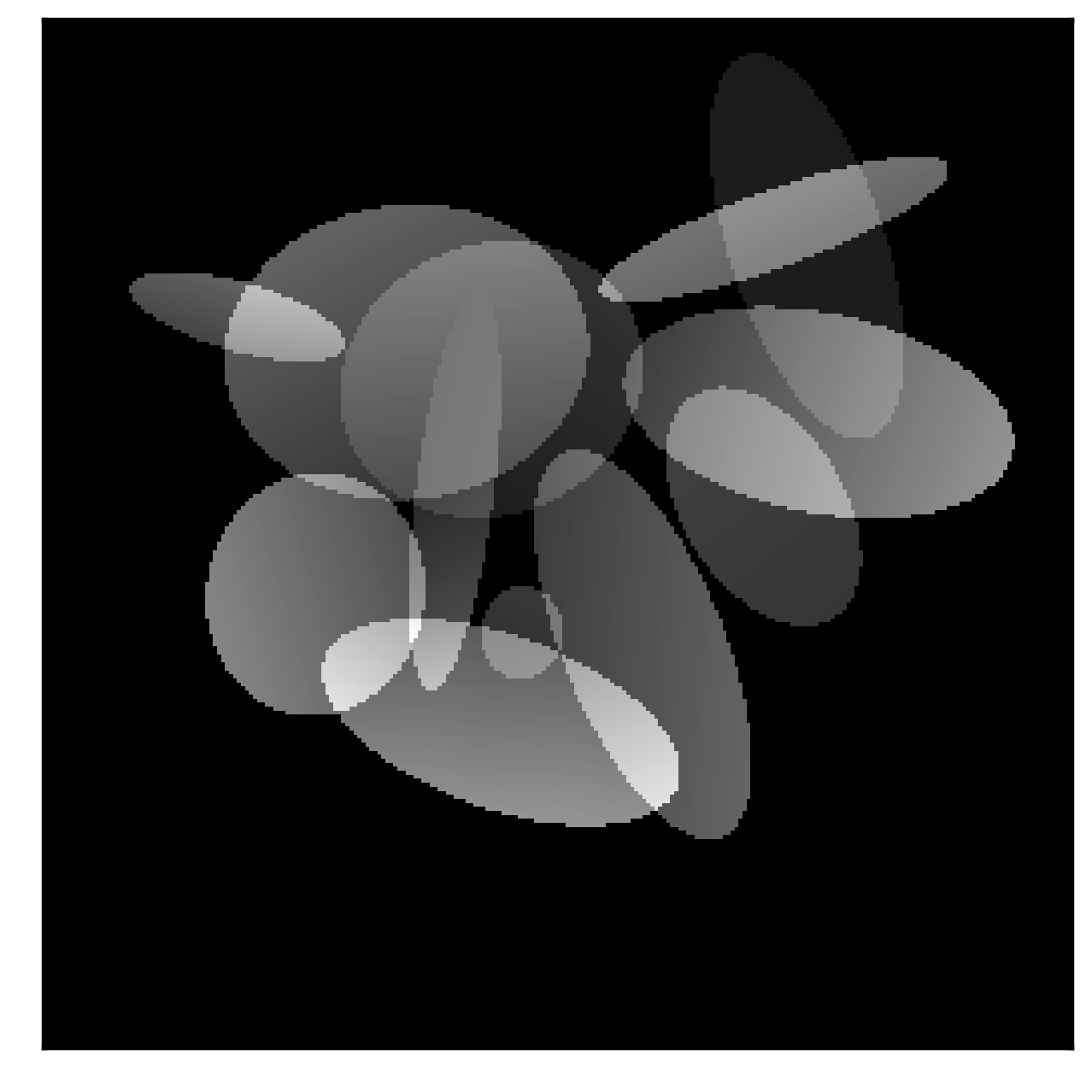}
\includegraphics[scale=0.65]{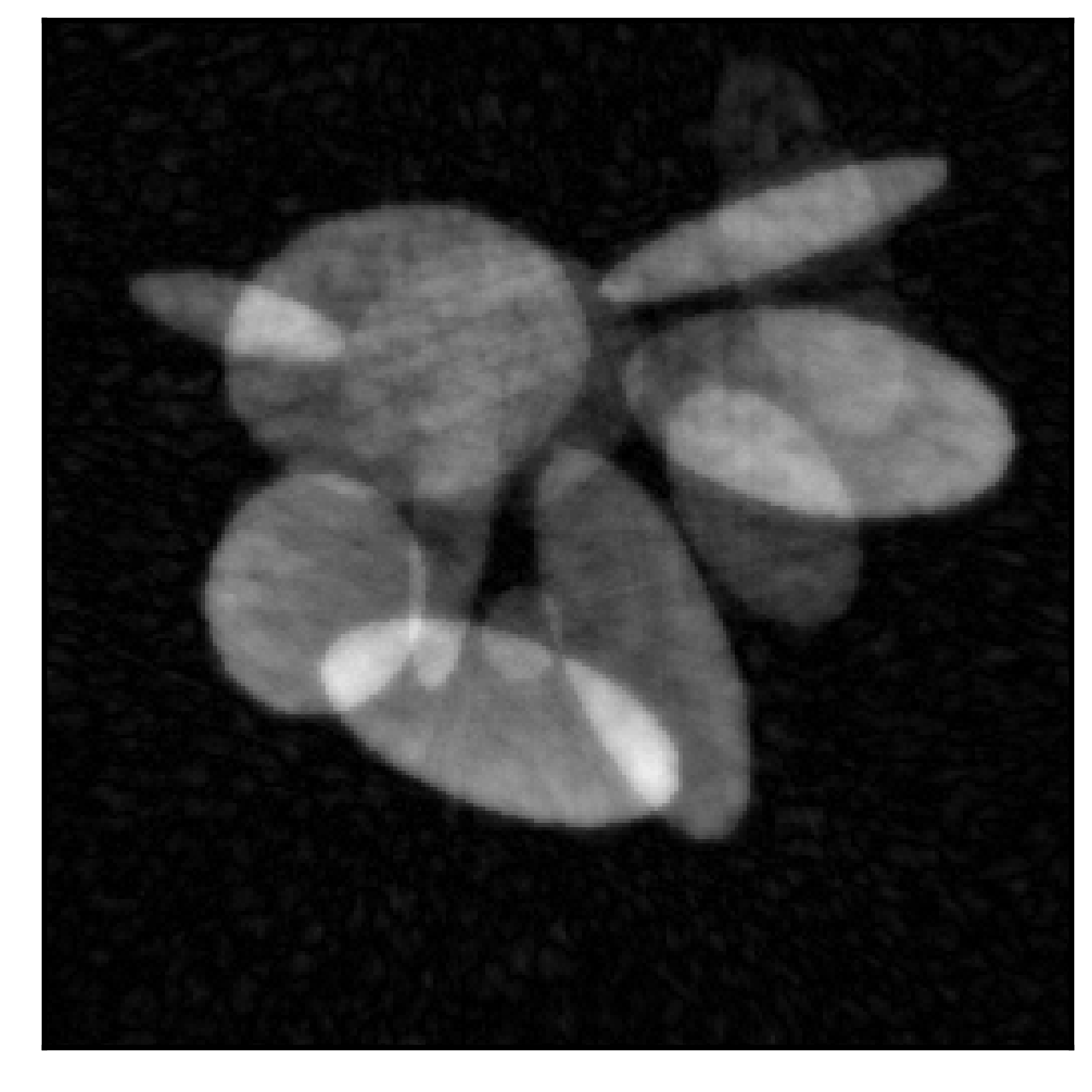}
\caption{Left: The sampling mask \(\mask\) used in the experiments, comprising 40 lines through the origin. The retained frequencies are shown in black. Middle: An example image \(x\). Right: A low quality linear reconstruction, using the pseudoinverse of the measurement matrix \(\mriopinv\mriop x\).}
\label{fig:mask_radial40}
\end{figure}

Compressed sensing (CS) refers to the approach of favoring solutions that are sparse under some given transform \(\Psi\).
Under certain conditions, the sparsest solution is also the one with the smallest \(\ell^1\)-norm \cite{foucart2013invitation}, and hence we are left with the following convex optimization problem:
\begin{equation}
	\label{eq:csconstrained}
	\min\limits_{z\in\spaceim}
	\norm{\Psi z}_1\ 
	\text{subject to}\ 
	\norm{\mriop z - y}_2 \leq \eta.
\end{equation}

Several choices of the transform \(\Psi\) can be found in the CS-MRI literature, including the identity, the wavelet transform, and the image gradient \cite{lu2008csmri}.
In this work, we focus on the last one of these choices.
More precisely, we choose \(\Psi=\nabla\), where \(\nabla\colon\spaceim\to\spaceim\times\spaceim\) is the two-dimensional finite difference operator on \(\spaceim\) with periodic boundary conditions.
The quantity \(\norm{\nabla z}_1\) is known as the \emph{total variation} (TV) of the image \(z\).
Minimization of the TV promotes
sparse gradients and hence 
piecewise constant solutions.
For computational efficiency, we solve the unconstrained formulation of (\ref{eq:csconstrained}),
\begin{equation}\label{eq:uncons}
	\reconeta{y}
	= \argmin\limits_{z\in\spaceim}\norm{\mriop z - y}_2^2
		+\lambda_\eta\norm{\nabla z}_1,
\end{equation}
where \(\lambda_\eta>0\) is a regularization parameter.
An appropriate choice of \(\lambda_\eta\) ensures that \(\reconeta{y}\) is a solution to (\ref{eq:csconstrained}) \cite{foucart2013invitation}.
However, since computing this exact value is infeasible, we choose \(\lambda_\eta\) based on a simple heuristic.

\subsection{DNNs as an alternative to CS}
Another approach to the MRI problem is to learn the inverse mapping \(\spacefr\to\spaceim\) from fully sampled data, i.e., to replace \(\recon\) by a neural network.
Several different strategies appear in the literature.
Fully learned networks learn the entire inversion without any knowledge of the forward model at all \cite{zhu2018image}, while others use the linear operator \(\mriopinv\colon\spacefr\to\spaceim\) as a first layer so that only the postprocessing step \(\spaceim\to\spaceim\) is learned \cite{wang2016accelerating}.
The forward model can also be incorporated at several stages in the network,
as in networks that are based on unrolled iterative optimization algorithms \cite{sun2016deep}.

For our experiments, we use the networks \(\tira\) and \(\tirafl\) from \cite{genzel2020solving},
which are based on the fully convolutional Tiramisu architecture \cite{jegou2017one}.
The difference between the two networks is that \(\tirafl\) is a fully learned network \(\spacefr\to\spaceim\), while \(\tira\) applies
the pseudoinverse \(\mriopinv\) to the measurements first.
Both networks are trained using the mean squared error as a loss function, and Gaussian noise is added to the input as a means of regularization.

\section{Adversarial perturbations}
\label{sec:adv}
The study of adversarial examples originates in image classification.
In that context, given an image and a classifier, an \emph{adversarial perturbation} is one that is imperceptible when added to the image but changes the output of the classifier.
The perturbed image is called an \emph{adversarial example}.
The imperceptibility of the perturbation is difficult to define, so this requirement is commonly replaced by a bound on the norm of the perturbation.
A popular choice is the \(\ell^\infty\)-norm  \cite{goodfellow2014explaining,madry2017towards},
since if each pixel in the image is changed only by a small value, then one can be sure that the semantic meaning of the image stays the same.
In contrast, the meaning of an image may be changed (say a handwritten 1 to a 7) by introducing a localized perturbation with a small \(\ell^2\)-norm.
It should be noted that there exist other adversarial image transformations that are not small in the \(\ell^\infty\)-norm, such as rotations, translations, or smooth deformations \cite{alaifari2019adef}.

Two difficulties arise when adapting the notion of adversarial examples to MRI.
Firstly, since measurements are in the frequency domain, the imperceptibility of perturbations is not very meaningful.
Secondly, since the output of the reconstruction method is a continuous variable (as opposed to discrete labels in classification),
a notion of severity is needed to quantify the effect of an adversarial perturbation.

It is natural to tackle the first problem by referring to the noise model; i.e., a perturbation is ``imperceptible'' if it is small in \(\ell^2\)-norm, as in both \cite{antun2020instabilities} and \cite{genzel2020solving}.
For the second problem, \cite{antun2020instabilities} inspects the visual quality of the reconstructed images, while \cite{genzel2020solving} uses the \(\ell^2\)-error for quantitative analysis of stability to perturbations.
Thus, for given measurements \(y\) and noise level \(\eta\), \cite{genzel2020solving} defines an adversarial perturbation \(e\) by
\begin{equation}
\label{eq:globalpert}
	e = \argmax\limits_{e\in\spacefr} \norm{\reconeta{y + e} - \reconeta{y}}_2
	\quad\text{such that}\quad
	\norm{e}_2 \leq \eta
\end{equation}
(and in fact, \cite{antun2020instabilities} solves  a similar unconstrained optimization problem).
We refer to
\begin{equation}
    \label{eq:artifact}
	\rho = \reconeta{y + e} - \reconeta{y}
\end{equation}
as the \emph{reconstruction artifact} induced by the perturbation \(e\).

In this work, we argue that the reconstruction \(\ell^2\)-error does not sufficiently capture the most harmful reconstruction artifacts.
In the medical setting, a misdiagnosis might be based on a localized anomaly in the image rather than an overall poor quality of reconstruction
(see Figure \ref{fig:mask_structpert}).
Thus we aim to create perturbations that cause localized reconstruction artifacts.
We replace the Euclidean norm in the objective with a weighted seminorm:
\begin{equation}
\label{eq:localpert}
	e = \argmax\limits_{e\in\spacefr} \norm{\phi\odot\left(\reconeta{y + e} - \reconeta{y}\right)}_2
	\quad\text{such that}\quad
	\norm{e}_2 \leq \eta
	,
\end{equation}
where \(\odot\) denotes entrywise product and \(\phi\in\spaceim\) is a weight vector.
To promote localized artifacts, we select a weight vector that is the (discrete) indicator function of a disk of radius \(\sigma>0\), centered at \(\mu=(\mu_1,\mu_2)\in[1,n]^2\).
In other words, we let \(\phi=\phi^{\mu,\sigma}\in\spaceim\) with
\begin{equation*}
\label{eq:diskweight}
\phi^{\mu,\sigma}_{ij}
=
\begin{cases}
1 \quad \text{if}\ (i-\mu_1)^2 + (j-\mu_2)^2\leq \sigma^2,\\
0 \quad \text{otherwise,}
\end{cases}
\end{equation*}
for \(i,j=1,\ldots,n\).

Although we target a specific location \(\mu\) in the image, we are interested in any large artifacts that may appear in reconstruction.
Therefore, we judge the severity of the reconstruction artifact \(\rho\) by its \(\ell^\infty\)-norm.
We solve (\ref{eq:localpert}) for all locations \(\mu\) on a regular grid in \([1,n]^2\) and select the perturbation \(e_\mu\) that maximizes
\begin{equation*}
	\norm{\reconeta{y + e_\mu} - \reconeta{y}}_\infty
\end{equation*}
as our adversarial perturbation.
The radius, \(\sigma\), is fixed at a small value relative to \(n\).
Note that for a large enough radius, (\ref{eq:localpert}) is equivalent to (\ref{eq:globalpert}).

\subsection{Visualizing perturbations}
While we search for reconstruction artifacts by perturbing the measurement vector \(y=\mriop x\), it is important to consider whether the artifact is in some sense already encoded in the perturbed measurements \(y+e\), rather than ``invented'' by the reconstruction method.
Indeed, in some cases shown in \cite{genzel2020solving} (see, for example, Fig.~7 therein) the same perturbation induces similar artifacts for both TV-regularized reconstruction and DNN-based reconstruction, indicating that these artifacts are present in the perturbation itself and are not created in the process of reconstruction.

To see if the perturbation \(e\in\spacefr\) encodes an artifact in this way, we visualize it by an image perturbation \(r\in\spaceim\) such that \(\mriop r = e\).
We select
\(
    r=\mriopinv e,
\)
as this perturbation is orthogonal to \(\ker\mriop\) 
and therefore does not contain any component that would be lost in the measurement process.
Moreover, since the Fourier transform is an isometry, we have \(\norm{r}_2=\norm{e}_2\).
Then we can compare the perturbed image \(x+r\) with the reconstruction
\(\reconeta{y+e}=\reconeta{\mriop(x+r)}\), both visually and in terms of the \(\ell^\infty\)-norm.

We note that by making use of \(\ker\mriop\), one can potentially select a different image perturbation \(r\), with a lower \(\ell^\infty\)-norm, that maps to the same perturbation \(e=\mriop r\).
Allowing such arbitrary image perturbations may change the comparison presented below, between TV-regularized reconstruction and DNN-based reconstruction.

\section{Results}
\label{sec:results}
We demonstrate the strategy described above by creating adversarial perturbations for synthetically generated \(256\times 256\) images of phantom ellipses (as in \cite{genzel2020solving}) with pixel values ranging from 0 to 1.
To generate the measurement vectors, we apply the MRI forward operator \(\mriop=\maskop\circ\fourier\) with a sampling mask \(\mask\) defined by 25, 40, or 80 lines through the origin (see Figure \ref{fig:mask_radial40}), which corresponds to \(11\%\), \(17\%\), or \(32\%\) of the coefficients.
We train the neural networks, \(\tira\) and \(\tirafl\), only on the data sampled on 40 lines.
 Moreover, all the figures in this paper are based on measurements using that same mask.

The reconstruction map in \eqref{eq:uncons}, \(\recon\), is realized by the alternating direction method of multipliers (ADMM) \cite{boyd2011distributed}, and the implementation is taken from \cite{genzel2020solving}.
For each noise level \(\eta\), the parameter \(\lambda_\eta\) is selected by grid search (along with the one free parameter of ADMM) to minimize the relative \(\ell^2\)-error of the reconstructed images compared with the original ellipse images, averaged over 50 samples.

We create adversarial perturbations for 50 ellipse images according to (\ref{eq:localpert}) at relative noise levels (\(\eta/\norm{y}_2\)) ranging from \(0.5\%\) to \(10\%\).  For each noise level and each image, we search for the best position of localization weight vector, \(\mu\), on an \(8\times 8\) grid.
The radius of the localization disk is fixed at \(\sigma=5\).
Experimentation showed that similar reconstruction artifacts appear for other values of \(\sigma\), as long as it is small.
For larger \(\sigma\), the artifacts are not well localized and resemble those of \cite{genzel2020solving}.

\begin{figure}
\centering
\includegraphics[scale=0.76]{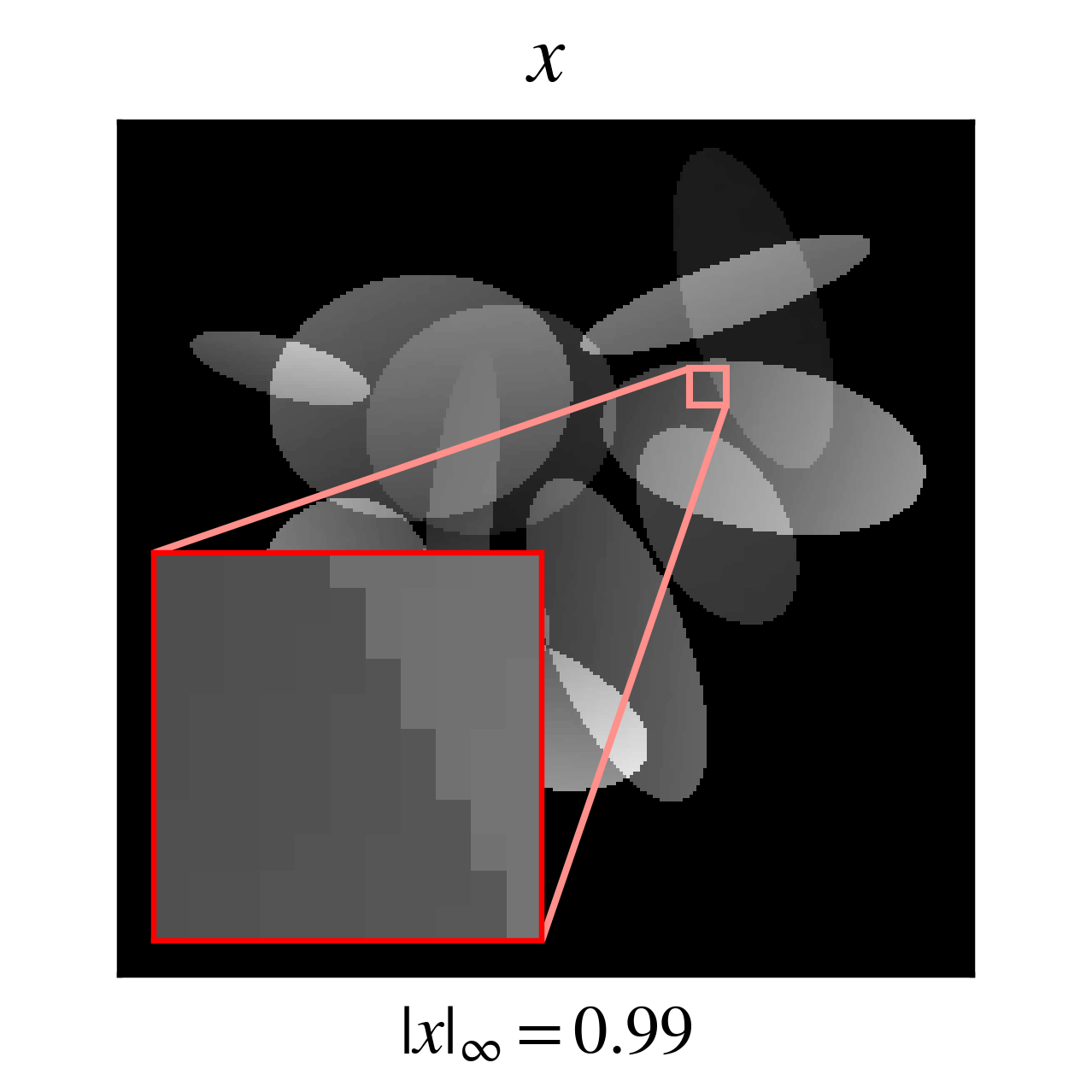}
\includegraphics[scale=0.76]{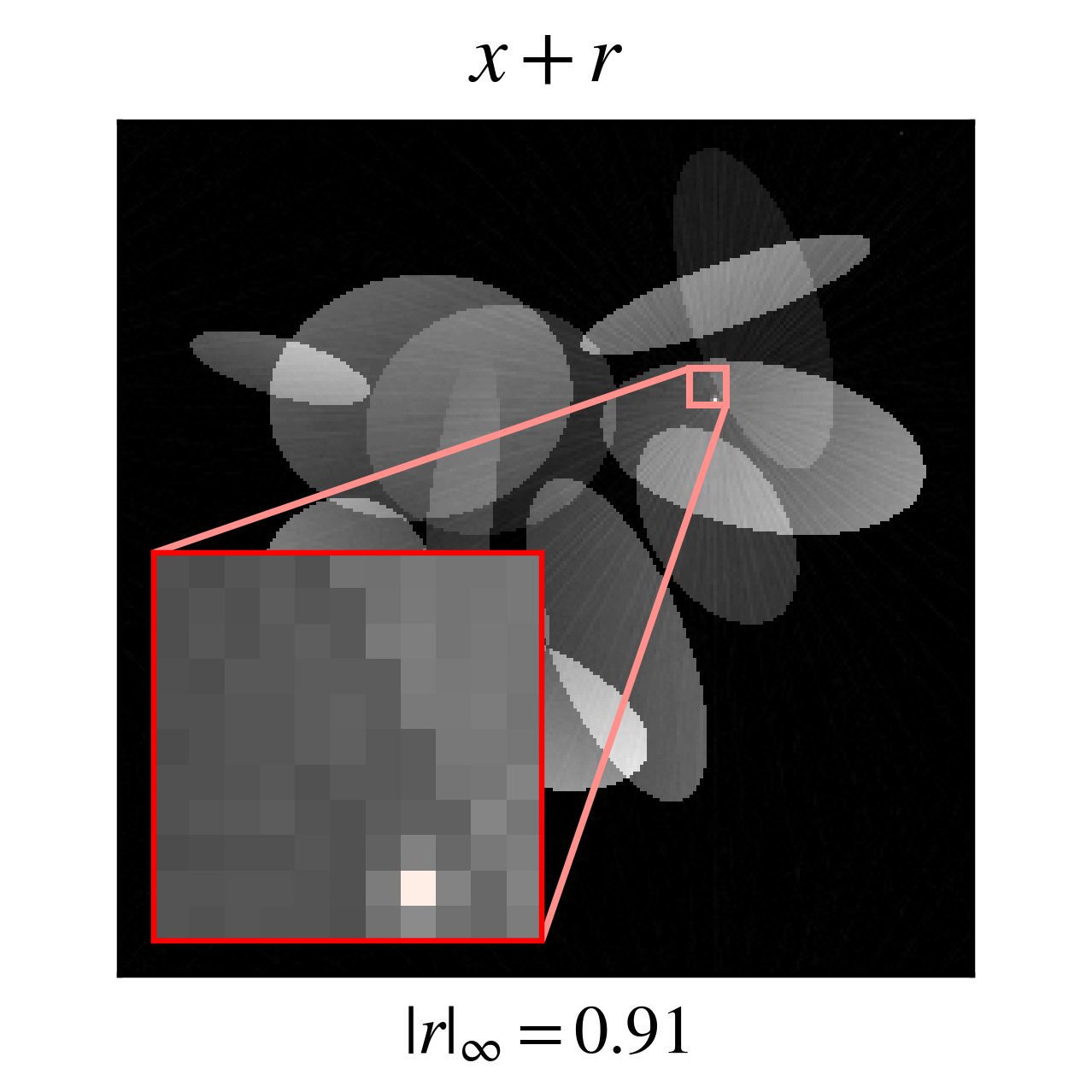}
\includegraphics[scale=0.76]{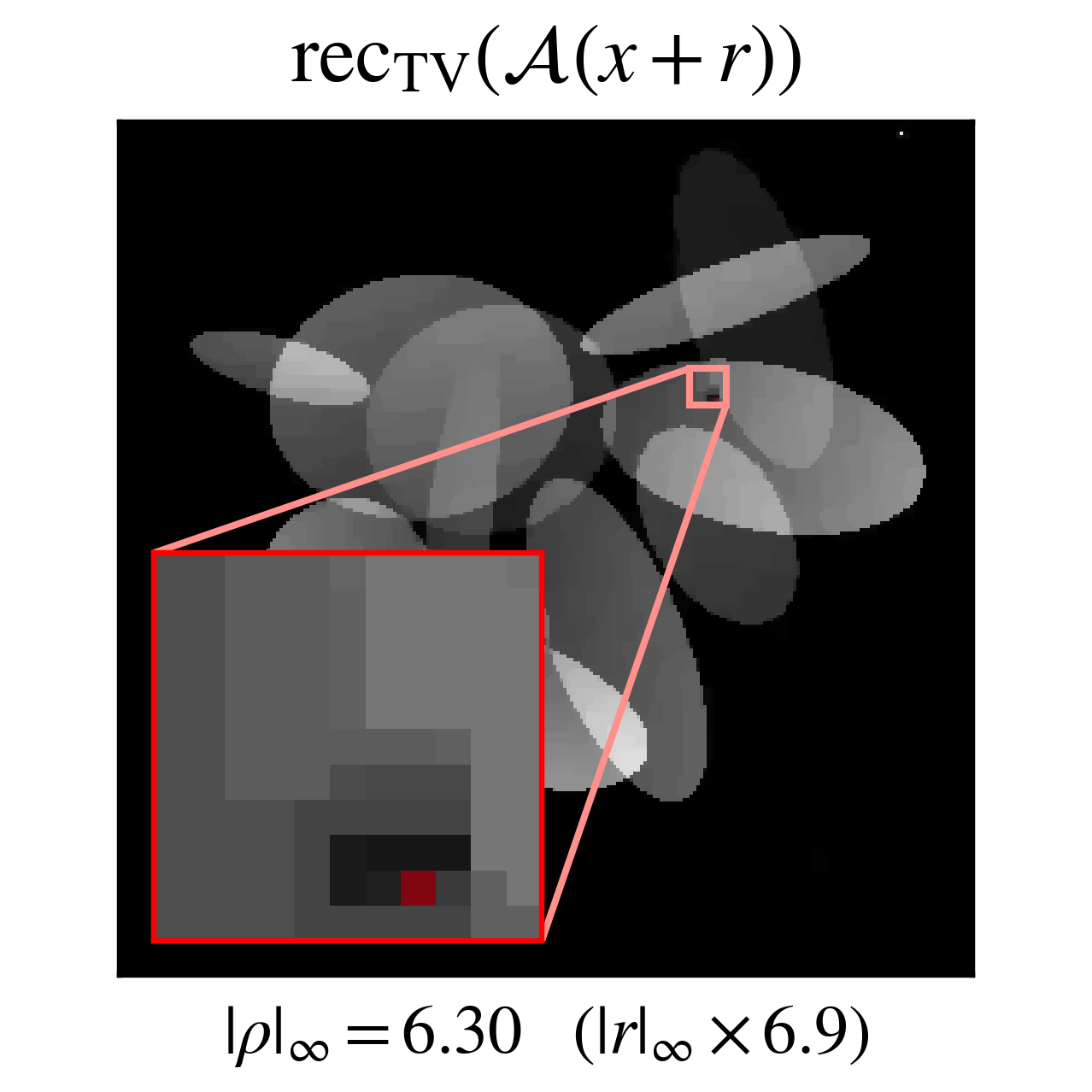}
\includegraphics[scale=0.76]{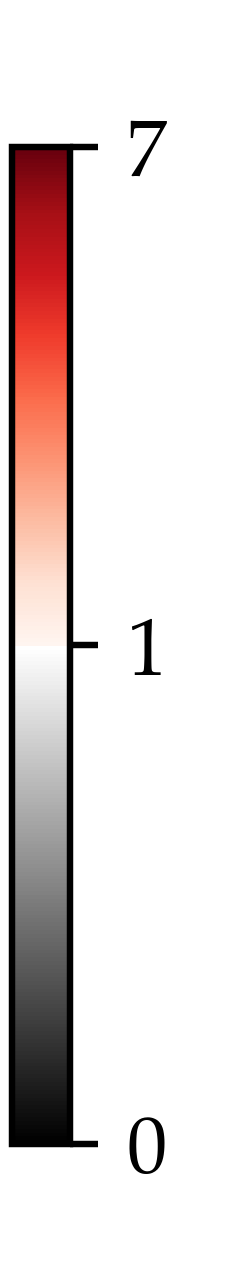}
\caption{
Left: The original image.
Middle: A perturbed image at 4\% relative noise. A close-up of the selected location, \(\mu\), shows that the image perturbation, \(r=\mriopinv e\), introduces a spike.
Right: The TV-reconstruction of the perturbed measurements. The spike is greatly amplified. Note that pixel values from 0 to 1 are shown in grayscale, while values beyond 1 are shown in red (the colorbar on the right is not linear).
}
\label{fig:simple_signal_nr4p}
\end{figure}

Figure \ref{fig:simple_signal_nr4p}
shows an example image \(x\), a perturbed image \(x+r\) (where \(r=\mriopinv e\) is a visual representation of the perturbation), and the reconstruction from perturbed measurements \(\reconeta{y+e}\).
The output of the reconstruction method is a vector in \(\spaceimdim{256}\), which we visualize by taking the entrywise modulus.
While the original ellipse images have pixel values between 0 and 1,
the reconstructed images have artifacts that reach far beyond that range.
Instead of transforming the reconstructed images to fit into \([0,1]\),
either by scaling by the reciprocal of the largest pixel value
or by replacing the values exceeding 1 with 1,
we display values between 0 and 1 in grayscale and values beyond 1 in shades of red.
This way,
details can be clearly seen and compared between images,
while the severity of tall artifacts can be assessed as well.

We now list some observations based on this quantitative experiment and visual inspection of examples.

\textbf{TV-regularized reconstruction amplifies artifacts already present in the perturbation.}
We observe that the adversarial perturbations indeed induce localized reconstruction artifacts in the form of a spike.
A spike is also noticeably present in the image perturbation, $r$, itself.
The image perturbation is therefore \emph{not} imperceptible in the \(\ell^\infty\)-sense, except at very low noise levels.
However, in \(\ell^\infty\)-norm, the reconstruction artifact \(\rho\) from (\ref{eq:artifact})
is much larger than the perturbation.
We consider the \emph{amplification factor}
\begin{equation*}
\alpha = \norm{\rho}_\infty/\norm{r}_\infty
\end{equation*}
to quantify this effect.
In Figure \ref{fig:simple_noise_lvls},
we observe a large \(\alpha\) at different noise levels, but visually, this phenomenon is especially conspicuous at higher noise levels.

\begin{figure}
\centering
\includegraphics[scale=0.58]{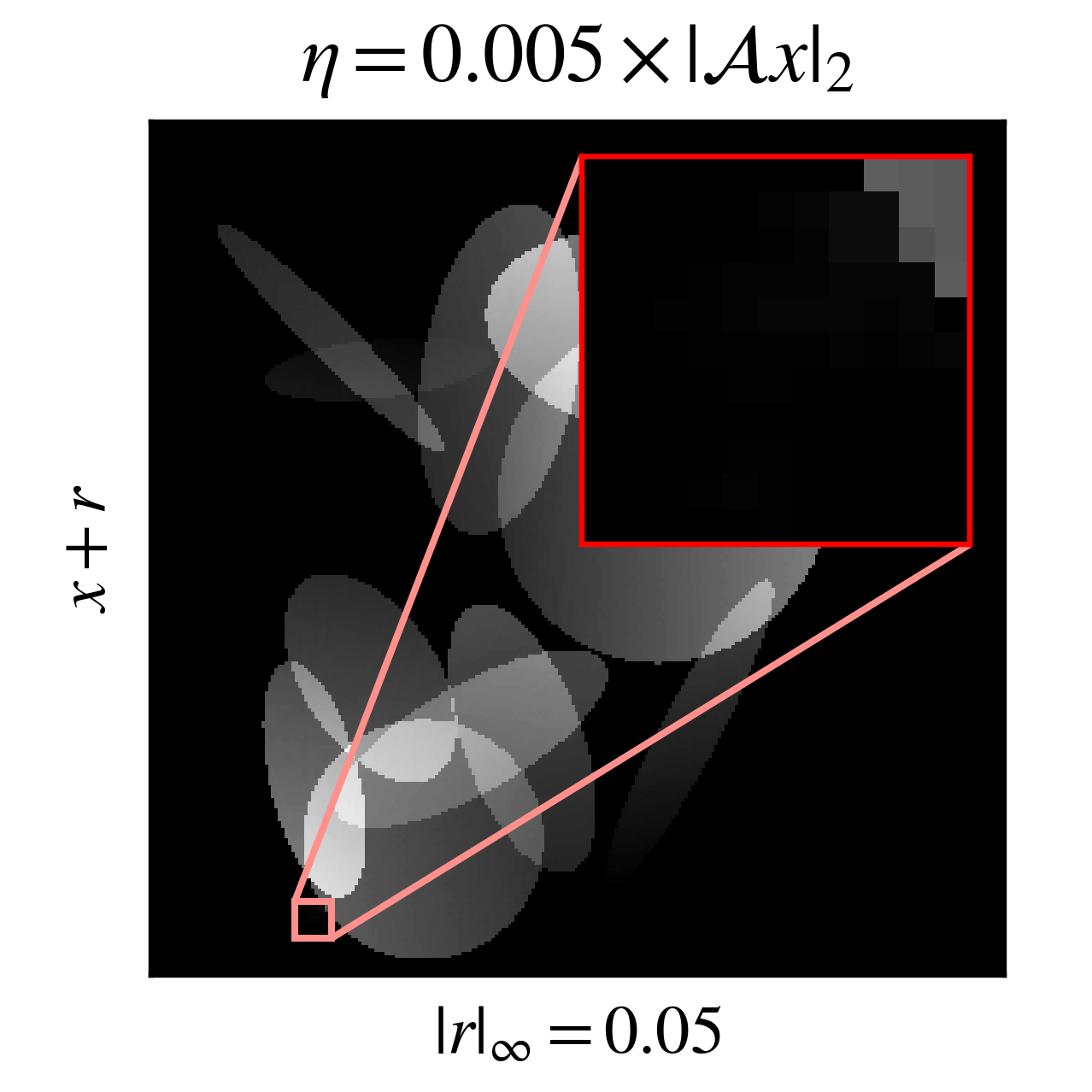}
\includegraphics[scale=0.58]{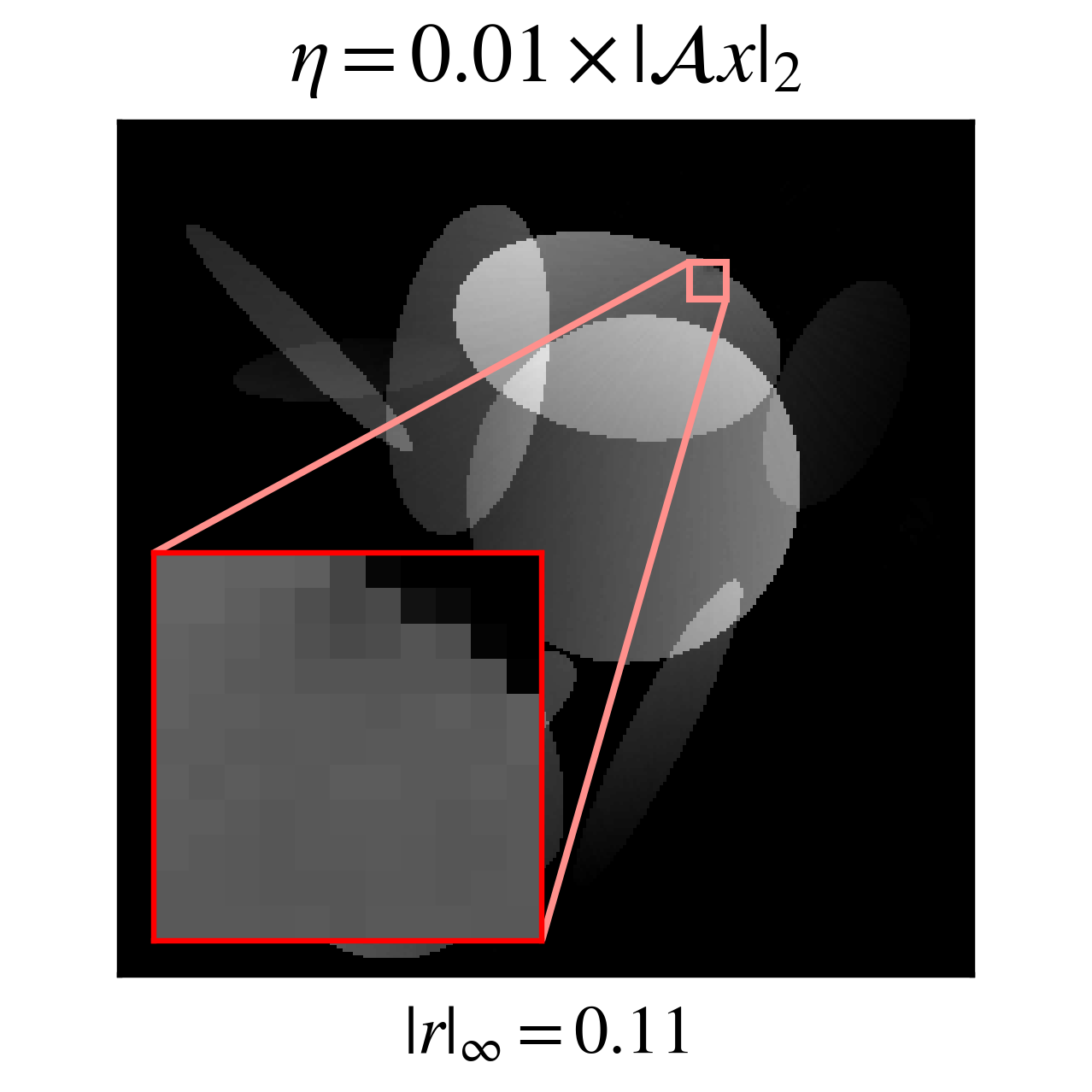}
\includegraphics[scale=0.58]{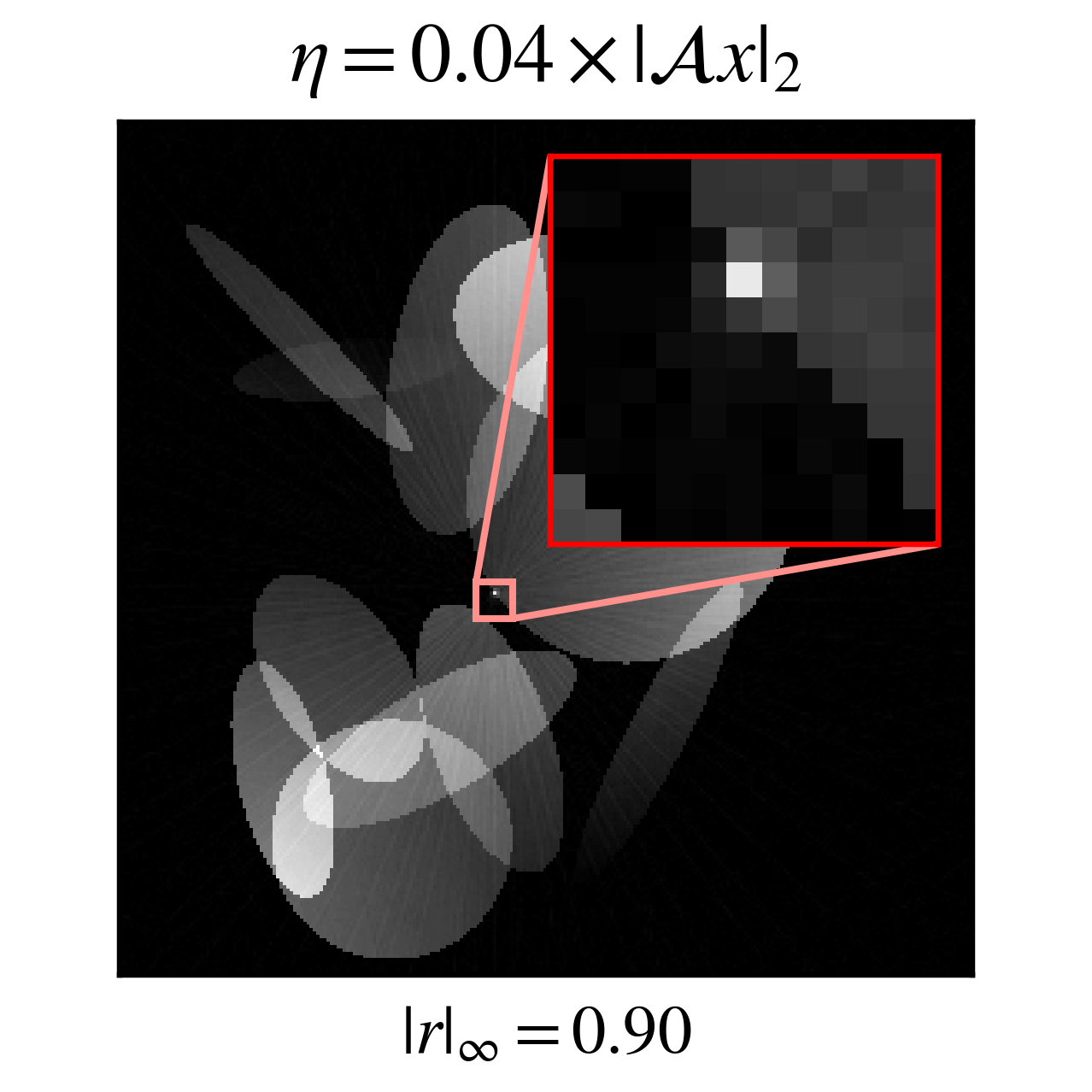}
\includegraphics[scale=0.58]{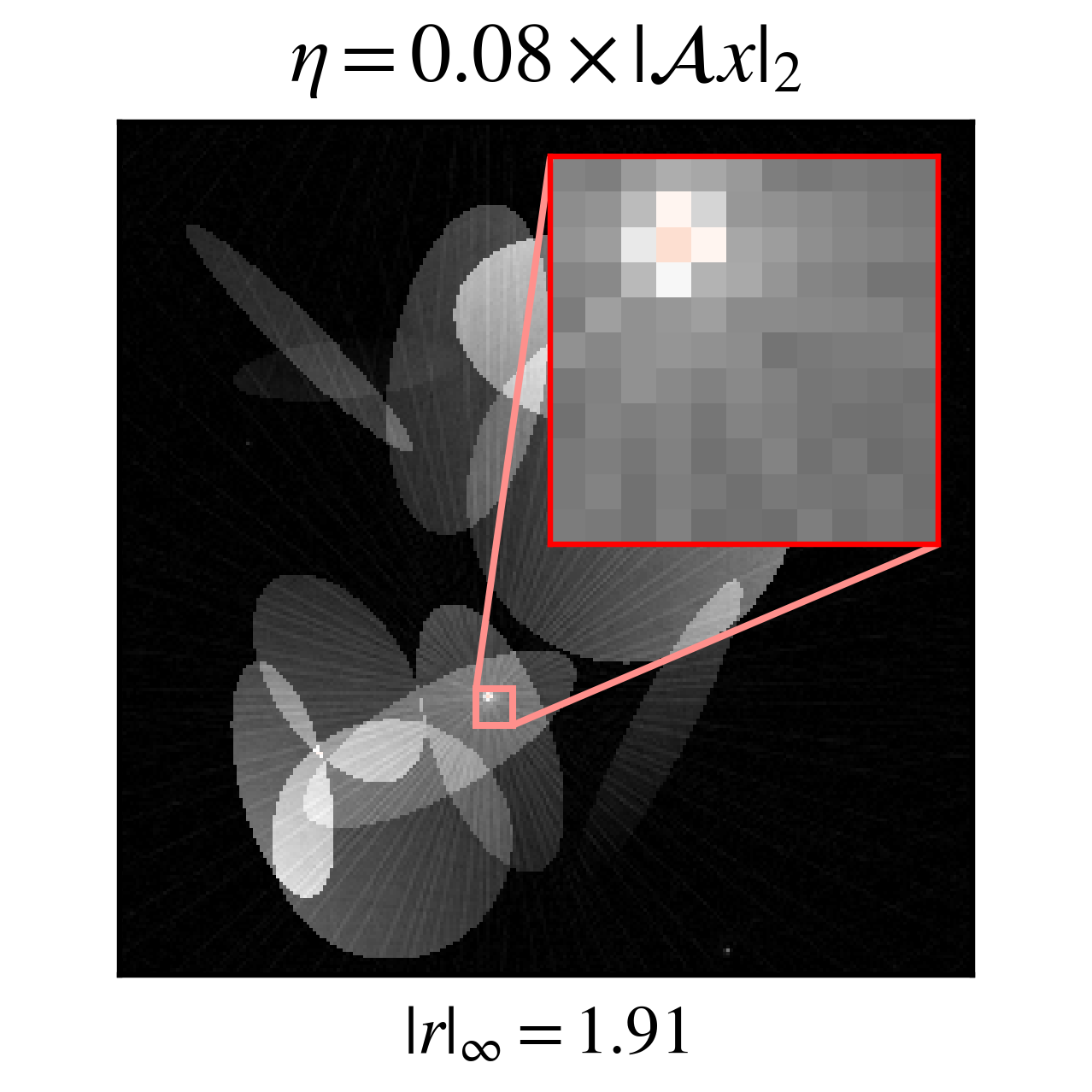}
\includegraphics[scale=0.58]{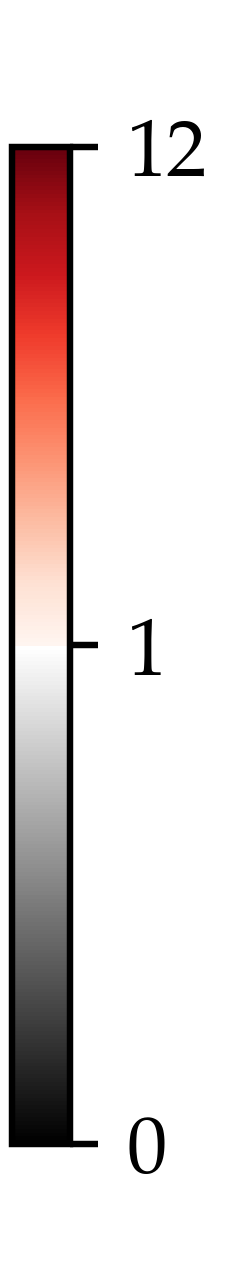}\\
\includegraphics[scale=0.58]{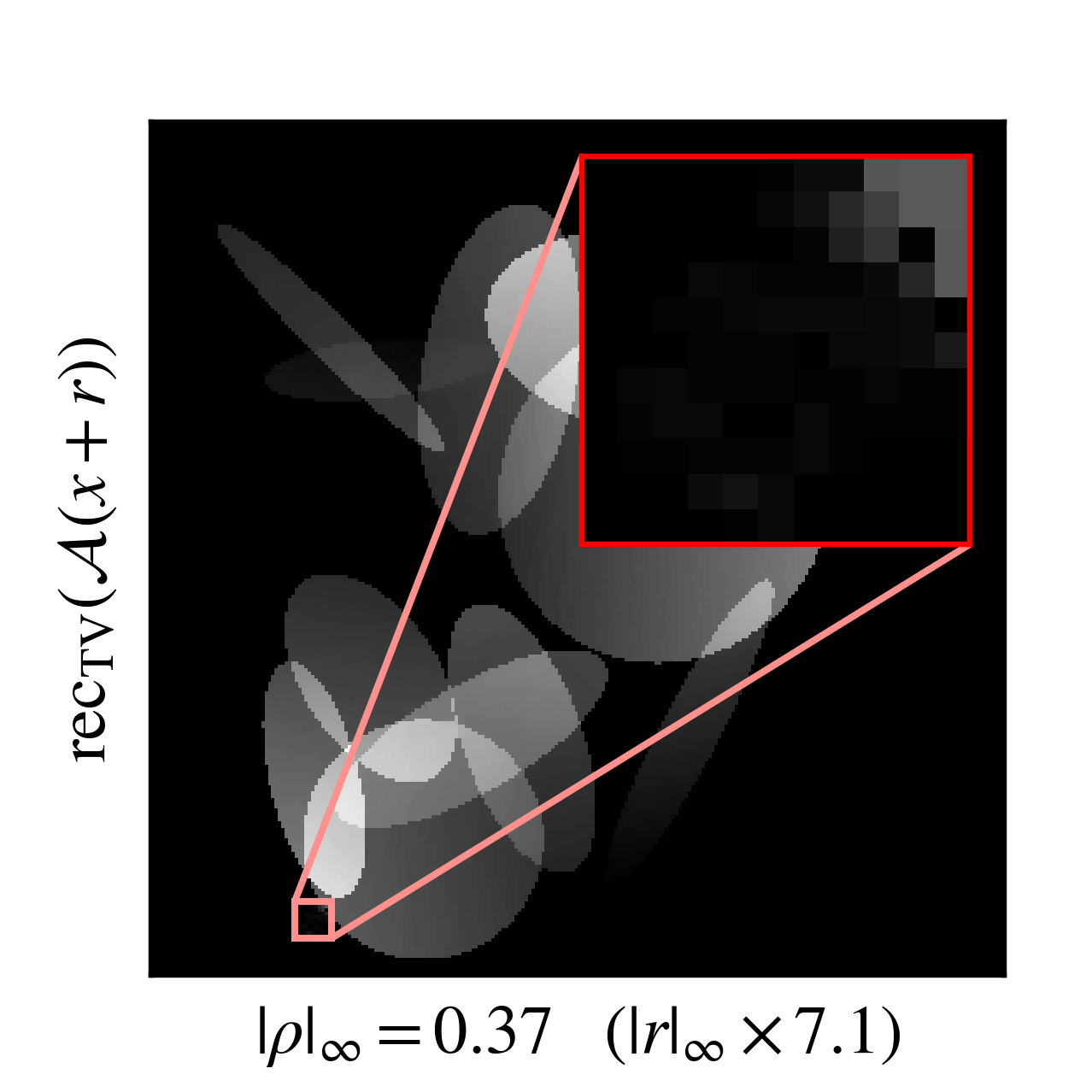}
\includegraphics[scale=0.58]{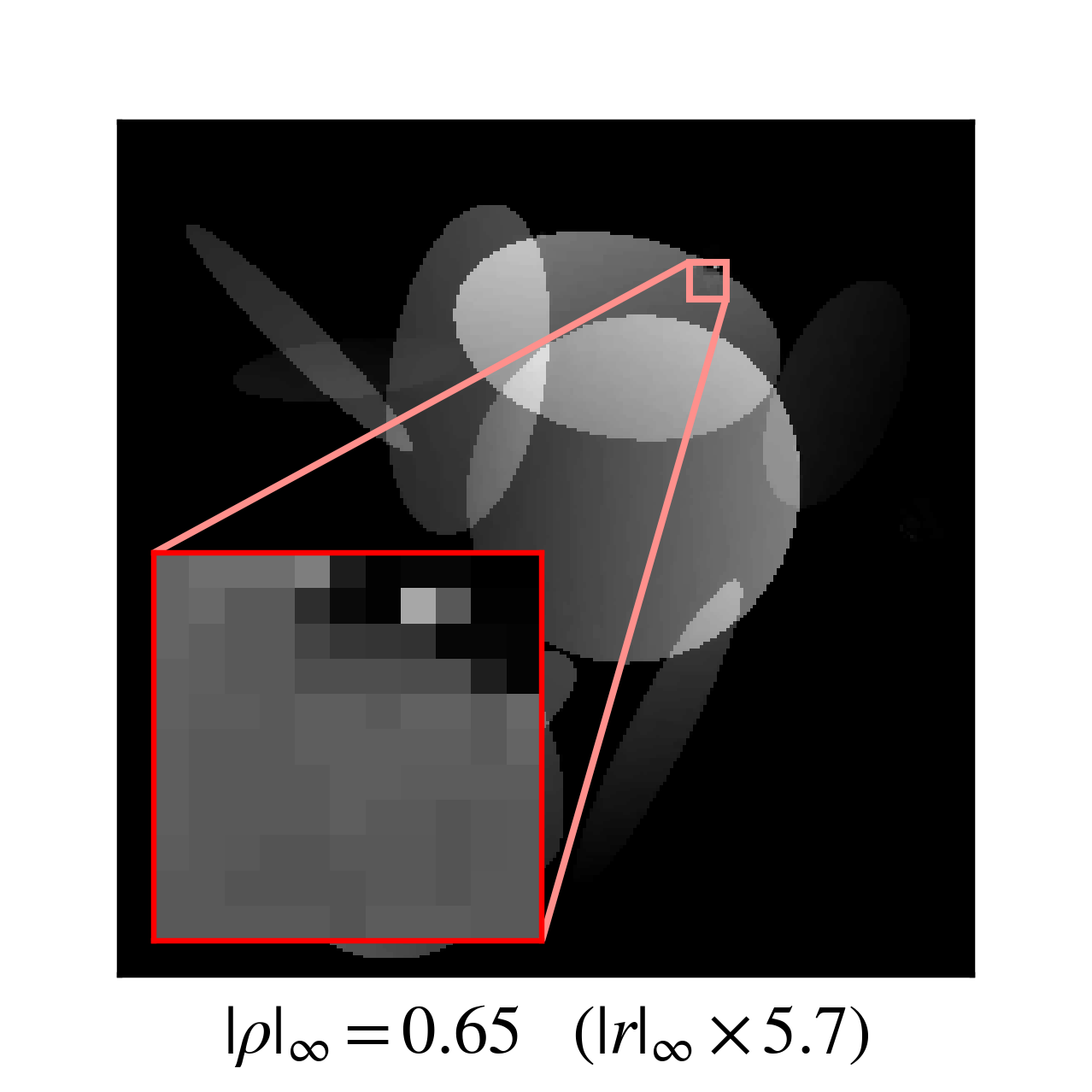}
\includegraphics[scale=0.58]{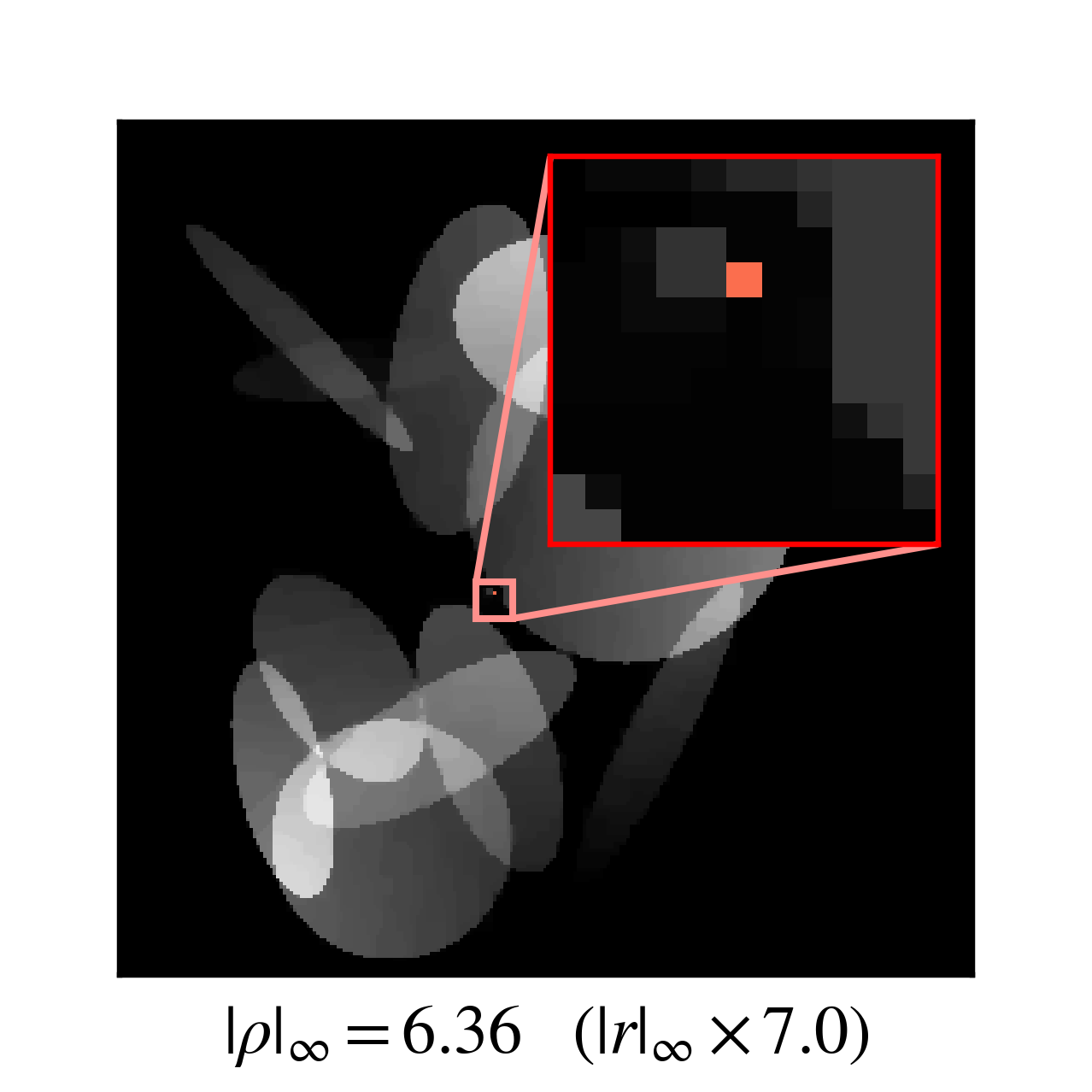}
\includegraphics[scale=0.58]{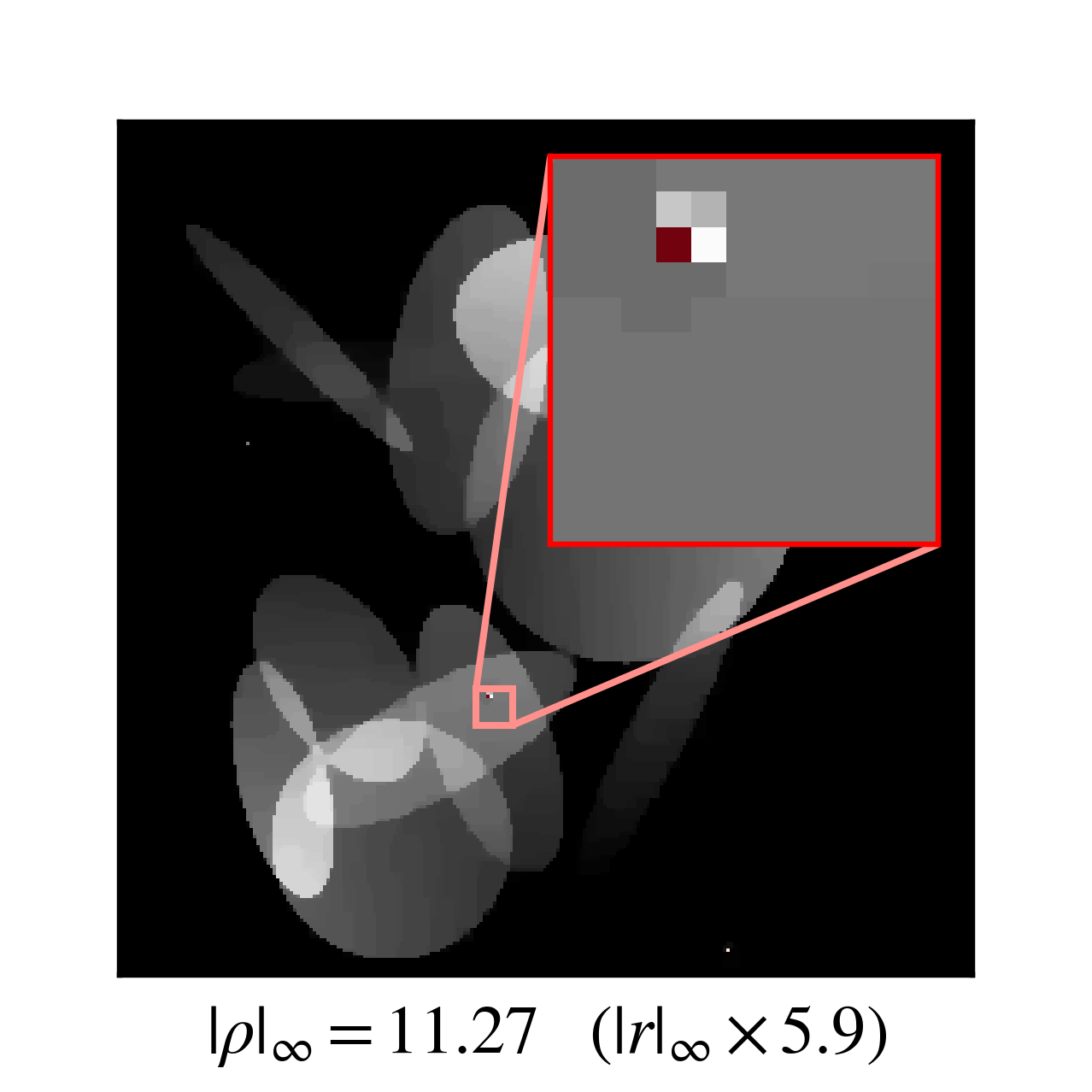}
\includegraphics[scale=0.58]{\FigPath exp_grid_attack_quant_radial40_width10_noise_lvls_cb }

\caption{Adversarial attacks at 0.5\%, 1\%, 4\%, and 8\% relative noise.
Upper row: the original image with the image perturbation added.
Lower row: TV-regularized reconstruction of the perturbed measurements. 
The reconstruction artifact \(\rho\) is consistently much larger than the perturbation \(r\) in the \(\ell^\infty\)-norm.
}
\label{fig:simple_noise_lvls}
\end{figure}

\textbf{The amplification factor is roughly equal to the subsampling factor.}
Table \ref{tab:quant} shows the average amplification factor of the 50 ellipse samples for each noise level.
Although the amplification factor varies more at low noise levels, for TV-regularization its average does not seem to depend on the noise level.
There is, however, a clear dependence on the number of coefficients in the measurements.
In fact, for reasons that will become clear in the next section, the average amplification factor is approximately the subsampling factor \(n^2/m\),
which for the 25-, 40-, and 80-line sampling masks is approximately \(9.4\), \(6.0\), and \(3.1\).

\begin{table}
\centering
\caption{
The amplification factor \(\alpha=\norm{\rho}_\infty/\norm{r}_\infty\) across different relative noise levels, averaged over 50 samples from the ellipse data set.
For \(\recon\), we reconstruct from measurements based on 25, 40, and 80 sampling lines in the frequency domain with the subsampling factor \(n^2/m\approx 9.4\), \(6.0\), and \(3.1\), respectively.
The DNNs are tested on the 40-line sampling mask only.
}
\vspace{0.5em}
\label{tab:quant}

\begin{tabular}{c | c c c | c c}
	\(\eta/\norm{y}_2\) &\(\recon\) (25 lines)&\(\recon\) (40 lines) &\(\recon\) (80 lines) &\(\tira\)& \(\tirafl\)\\
\hline
0.5\%	& \(8.12	\pm 2.83\) & \(6.23	\pm 2.16\) & \(3.15	\pm 0.15\) & \(7.29	\pm 5.00\) &\(7.36	\pm 3.83\)\\
1.0\%	& \(9.05	\pm 2.56\) & \(5.64	\pm 2.24\) & \(3.55	\pm 0.26\) & \(5.91	\pm 4.76\) &\(5.67	\pm 4.40\)\\
2.0\%	& \(9.05	\pm 2.40\) & \(5.88	\pm 0.40\) & \(3.22	\pm 0.19\) & \(3.20	\pm 1.53\) &\(3.59	\pm 2.61\)\\
4.0\%	& \(9.06	\pm 0.43\) & \(6.88	\pm 0.56\) & \(3.19	\pm 0.10\) & \(2.38	\pm 0.81\) &\(2.34	\pm 0.69\)\\
8.0\%	& \(9.75	\pm 0.52\) & \(5.91	\pm 0.28\) & \(3.07	\pm 0.07\) & \(1.89	\pm 0.44\) &\(1.96	\pm 0.58\)\\    
10.0\%	& \(8.80	\pm 0.71\) & \(5.65	\pm 0.42\) & \(2.99	\pm 0.11\) & \(1.80	\pm 0.48\) &\(1.77	\pm 0.60\)\\
\end{tabular}
\end{table}

\textbf{The amplification factor is smaller for DNNs.}
Applying the attack strategy \eqref{eq:localpert} to the neural networks \(\tira\) and \(\tirafl\) does produce reconstruction artifacts.
However, at \(>1\%\) noise, the attack is far less effective than for \(\recon\).
Figure \ref{fig:simple_signal_dnn_nr4p} shows a typical artifact created by \(\tira\), which has a much lower \(\ell^\infty\)-norm than those created by \(\recon\).
It should be noted that these DNNs have been trained with jittering and that this influences  their ability to recover detail.
Therefore a different amplification factor should be expected for networks trained at a lower noise level.

\begin{figure}
\centering
\includegraphics[scale=0.76]{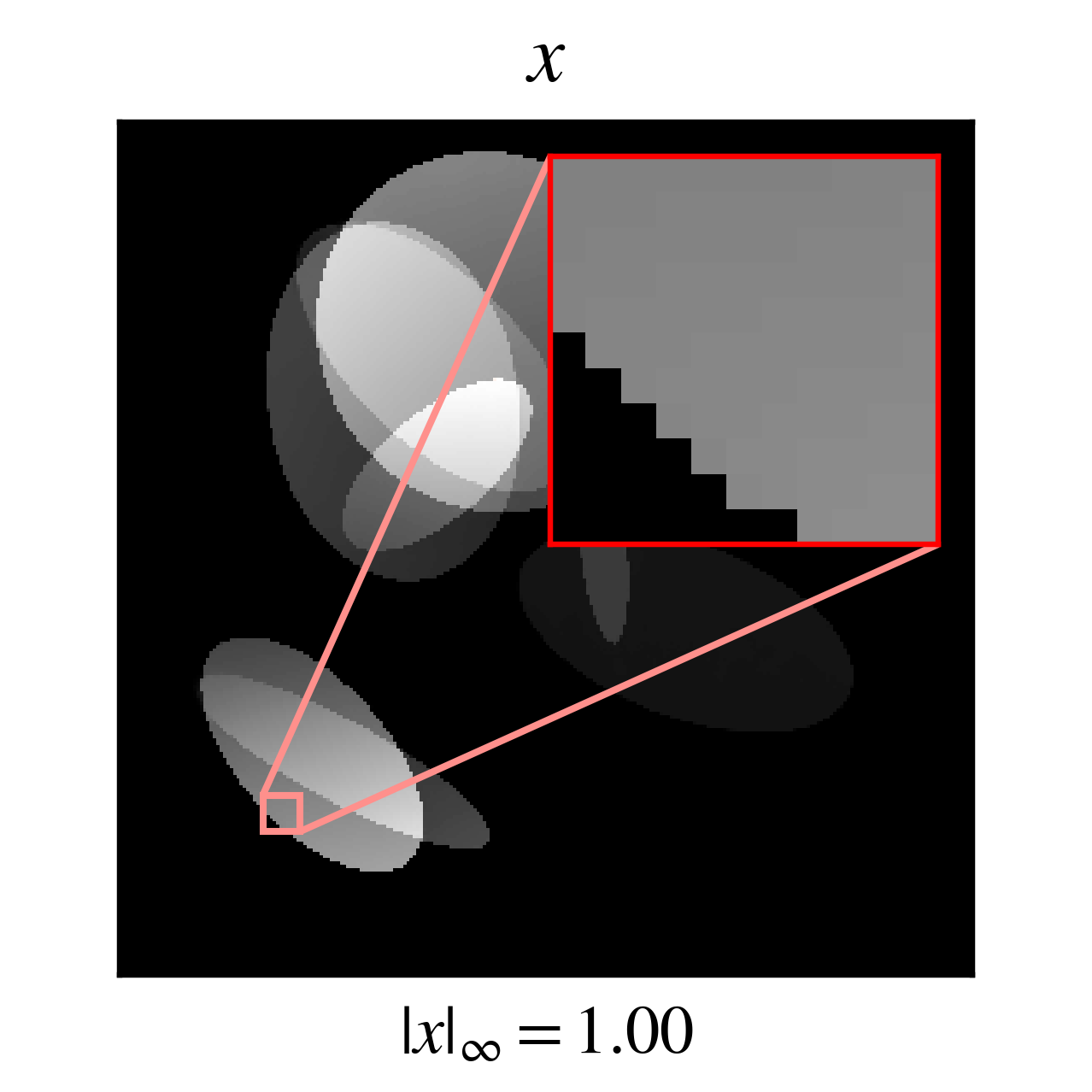}
\includegraphics[scale=0.76]{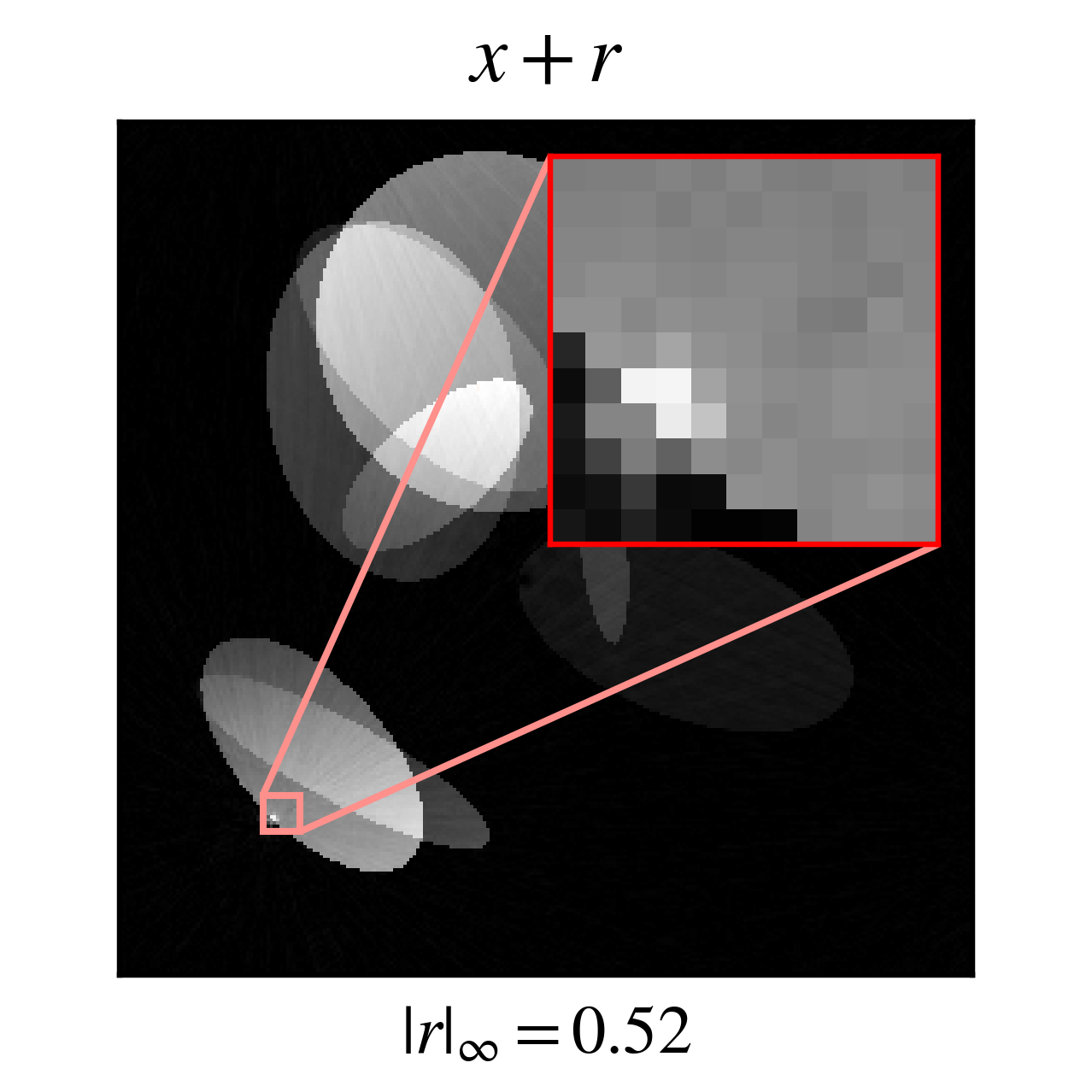}
\includegraphics[scale=0.76]{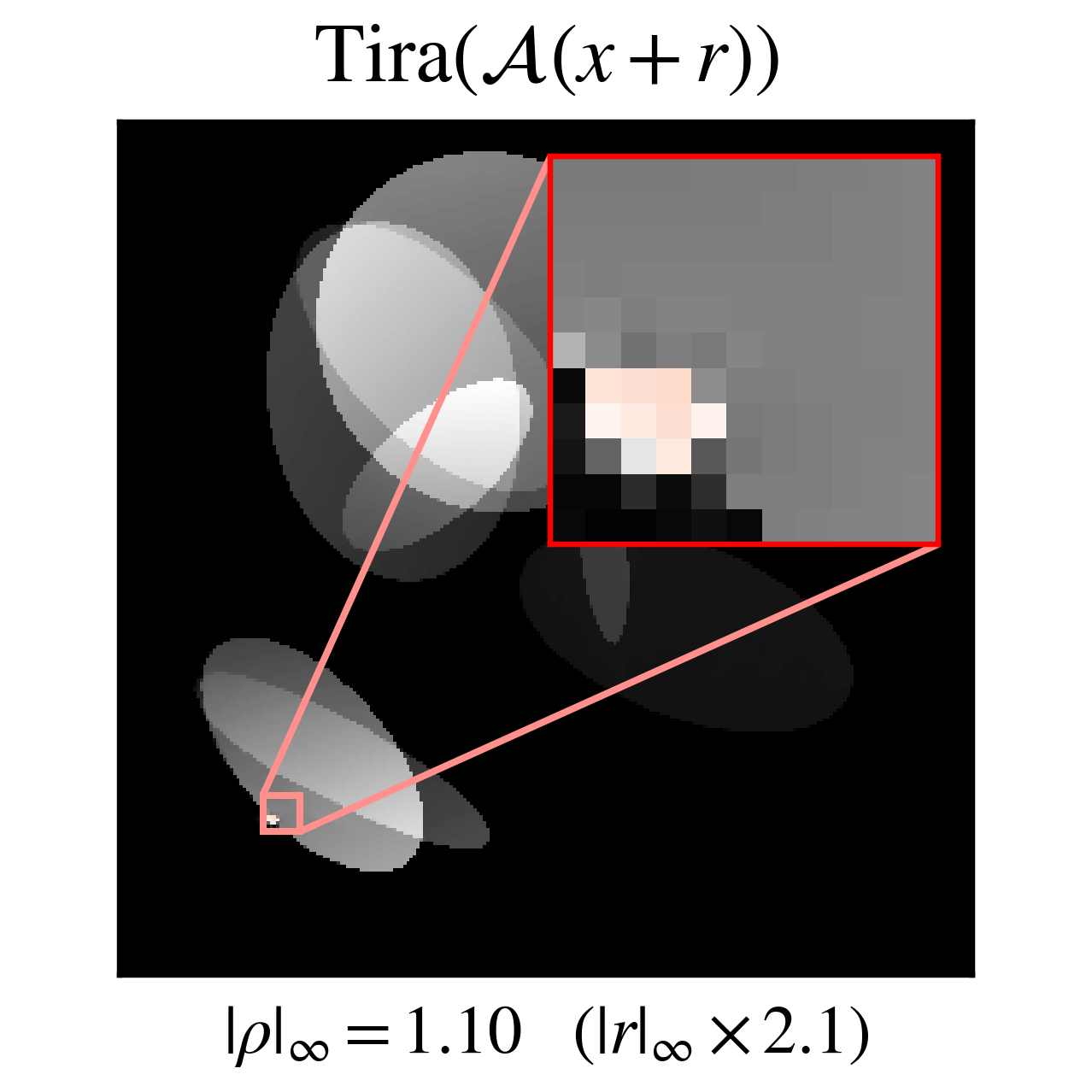}
\includegraphics[scale=0.76]{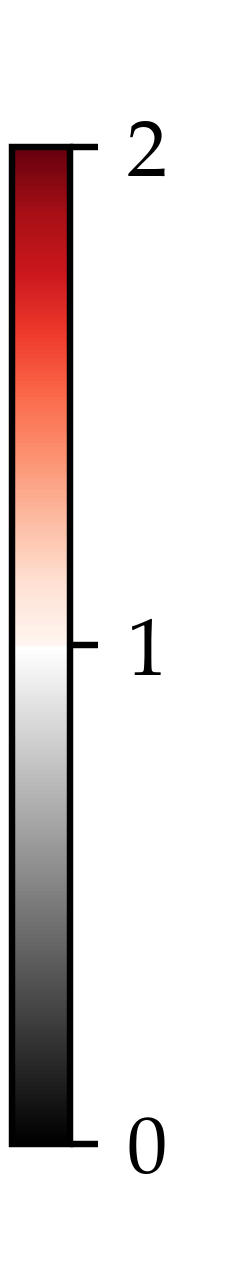}
\caption{
Analogous to Figure~\ref{fig:simple_signal_nr4p}, with the reconstruction method \(\recon\) replaced by the DNN \(\tira\).
Here, the amplification factor is \(\alpha = 2.13\).
}
\label{fig:simple_signal_dnn_nr4p}
\end{figure}

\textbf{Perturbations made for DNN transfer to \(\recon\) but not vice versa.}
In Figure~\ref{fig:simple_transfer_tira_nr4p}, we see that a perturbation created for a DNN can transfer to all three reconstruction methods, with \(\recon\) showing the most severe artifact.
On the other hand, using the DNNs to reconstruct from \(y+e\), where \(e\) is a perturbation created for TV-regularized reconstruction, produces images of good quality and no visible artifacts (Figure~\ref{fig:simple_transfer_tv_nr4p}).
In fact, the peak of the image perturbation \(r\) is dampened by the DNN reconstruction.

\begin{figure}
\centering
\includegraphics[scale=0.58]{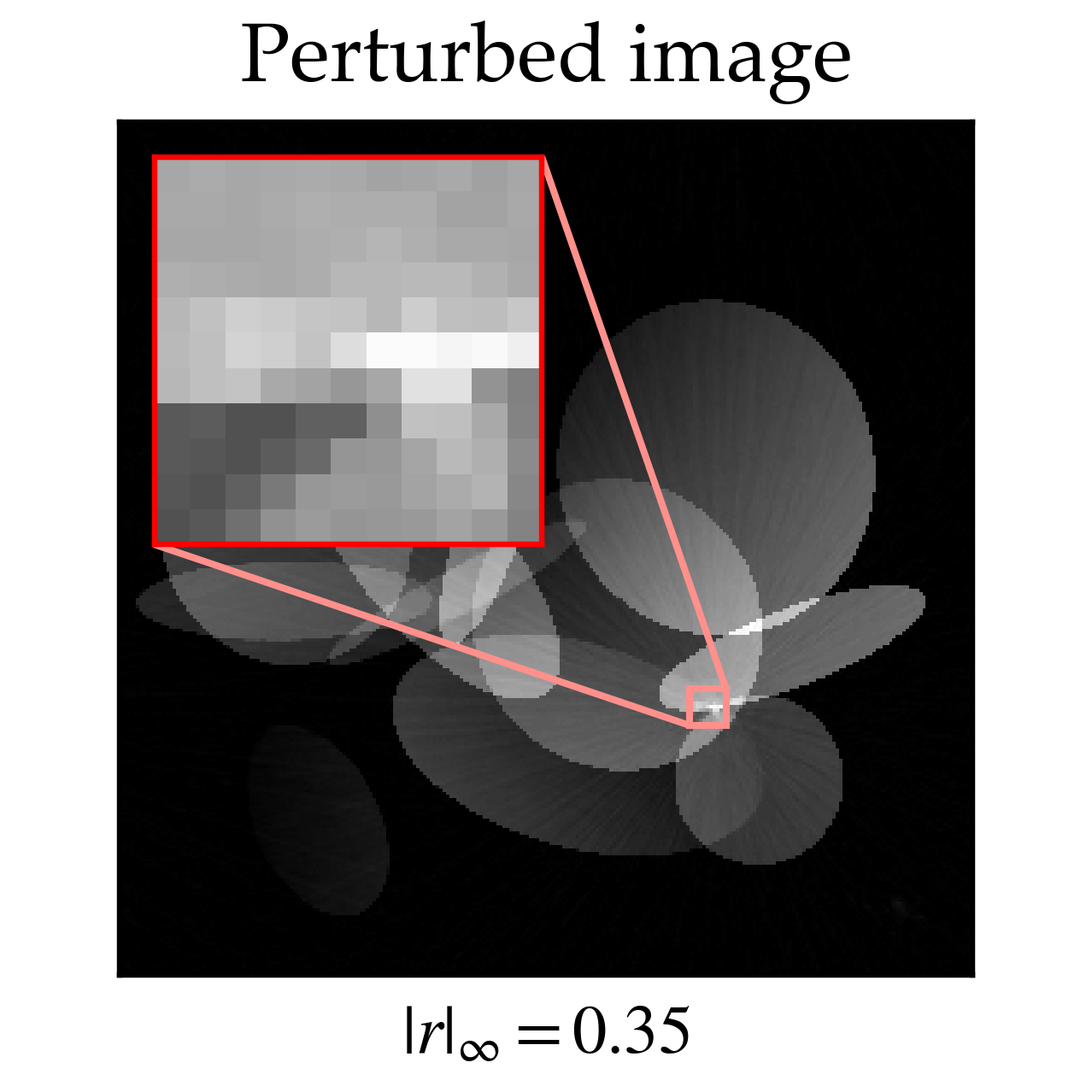}
\includegraphics[scale=0.58]{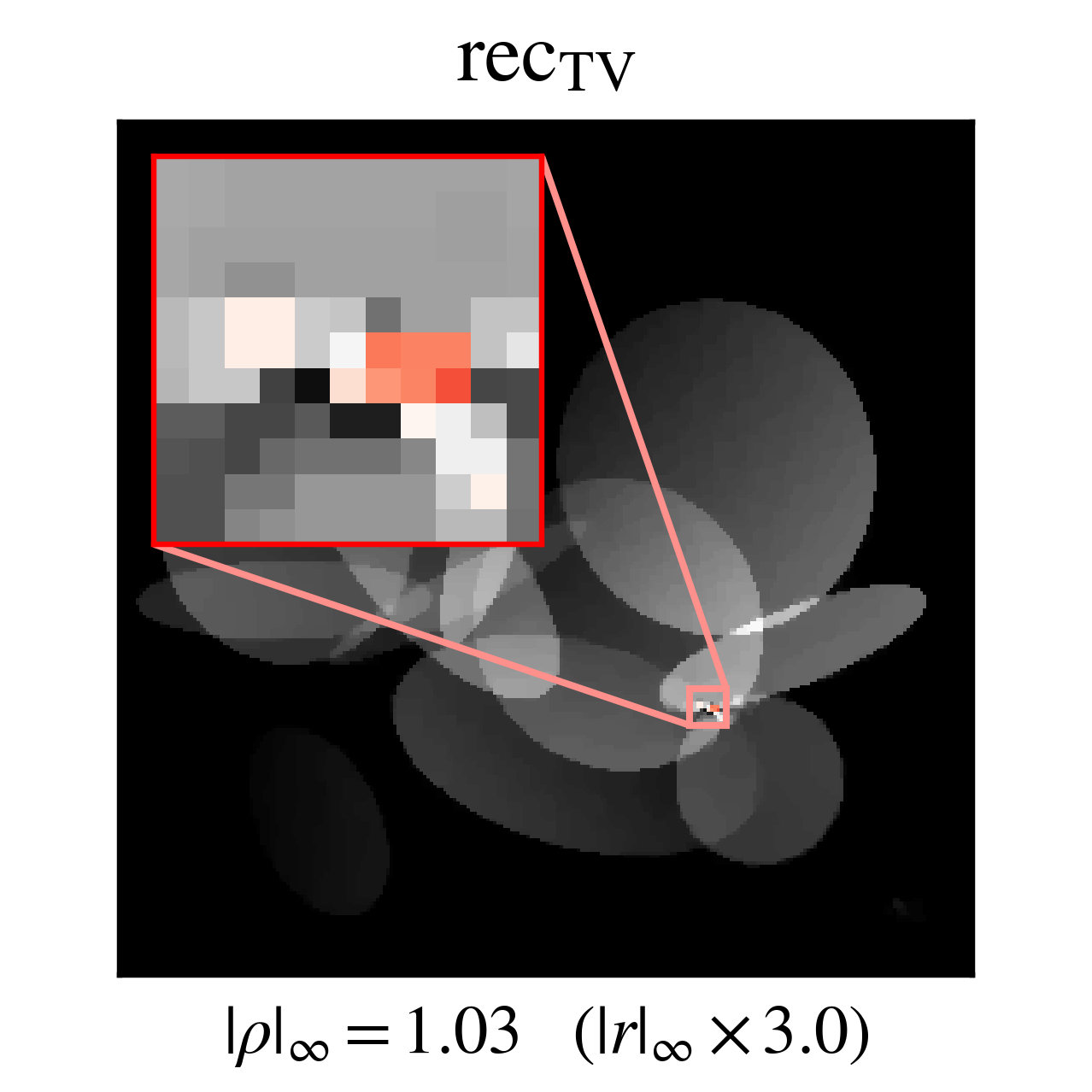}
\includegraphics[scale=0.58]{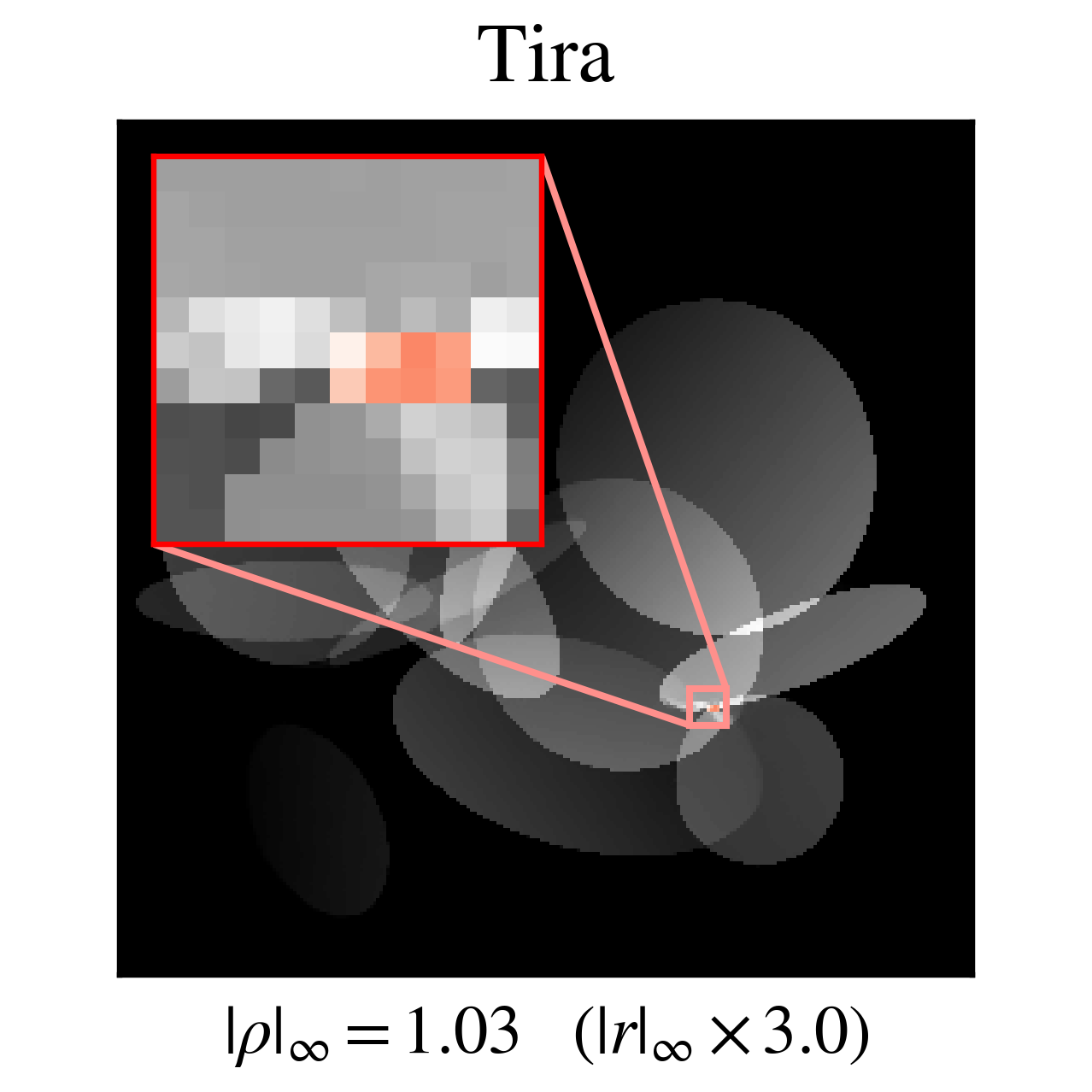}
\includegraphics[scale=0.58]{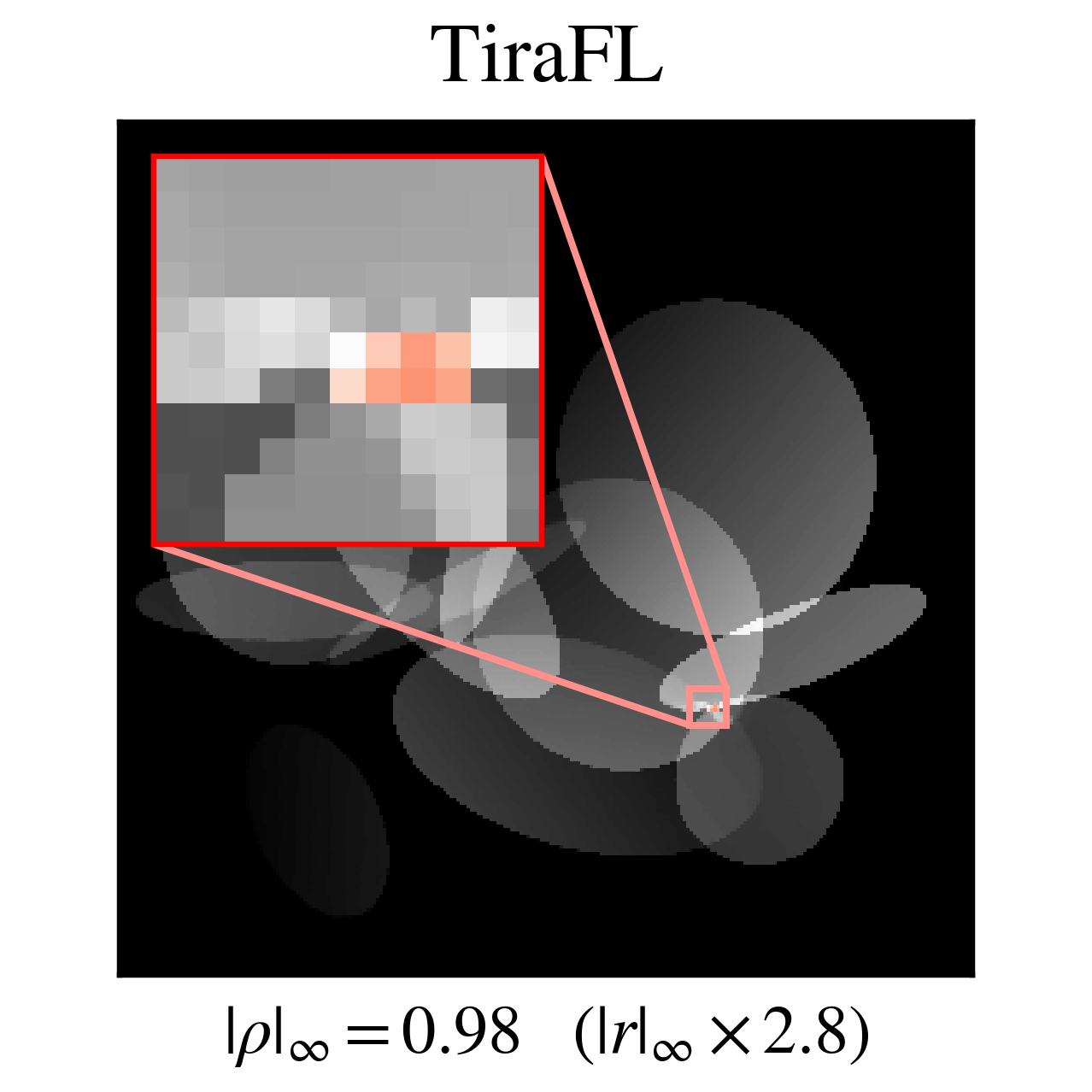}
\includegraphics[scale=0.58]{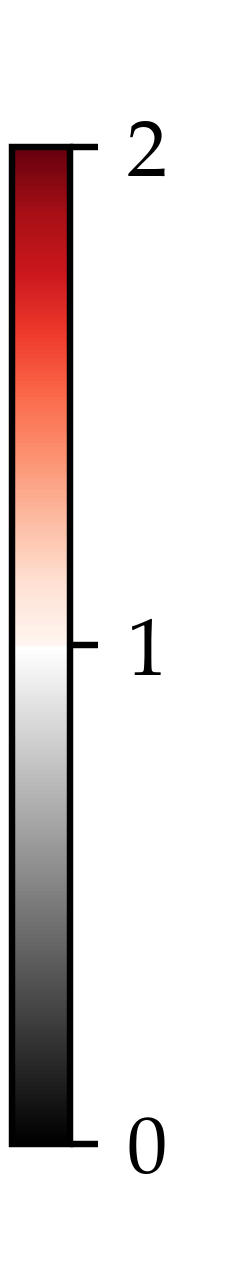}
\caption{
An adversarial perturbation is created with respect to \(\tira\) at a 4\% relative noise level, and then \(\recon\), \(\tira\), and \(\tirafl\) are used to reconstruct from the perturbed measurements.
All methods show similar reconstruction artifacts.
}
\label{fig:simple_transfer_tira_nr4p}
\end{figure}

\begin{figure}
\centering
\includegraphics[scale=0.58]{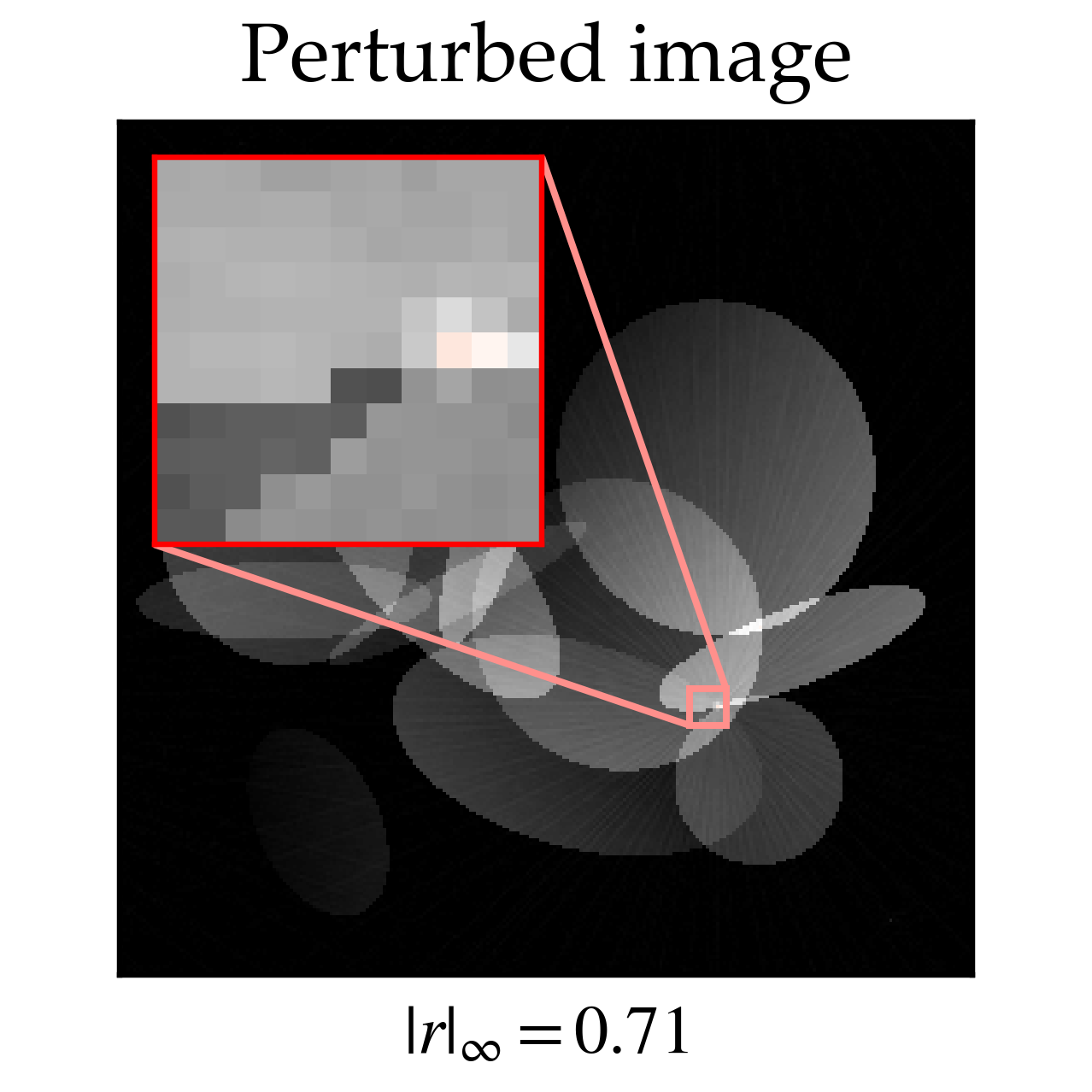}
\includegraphics[scale=0.58]{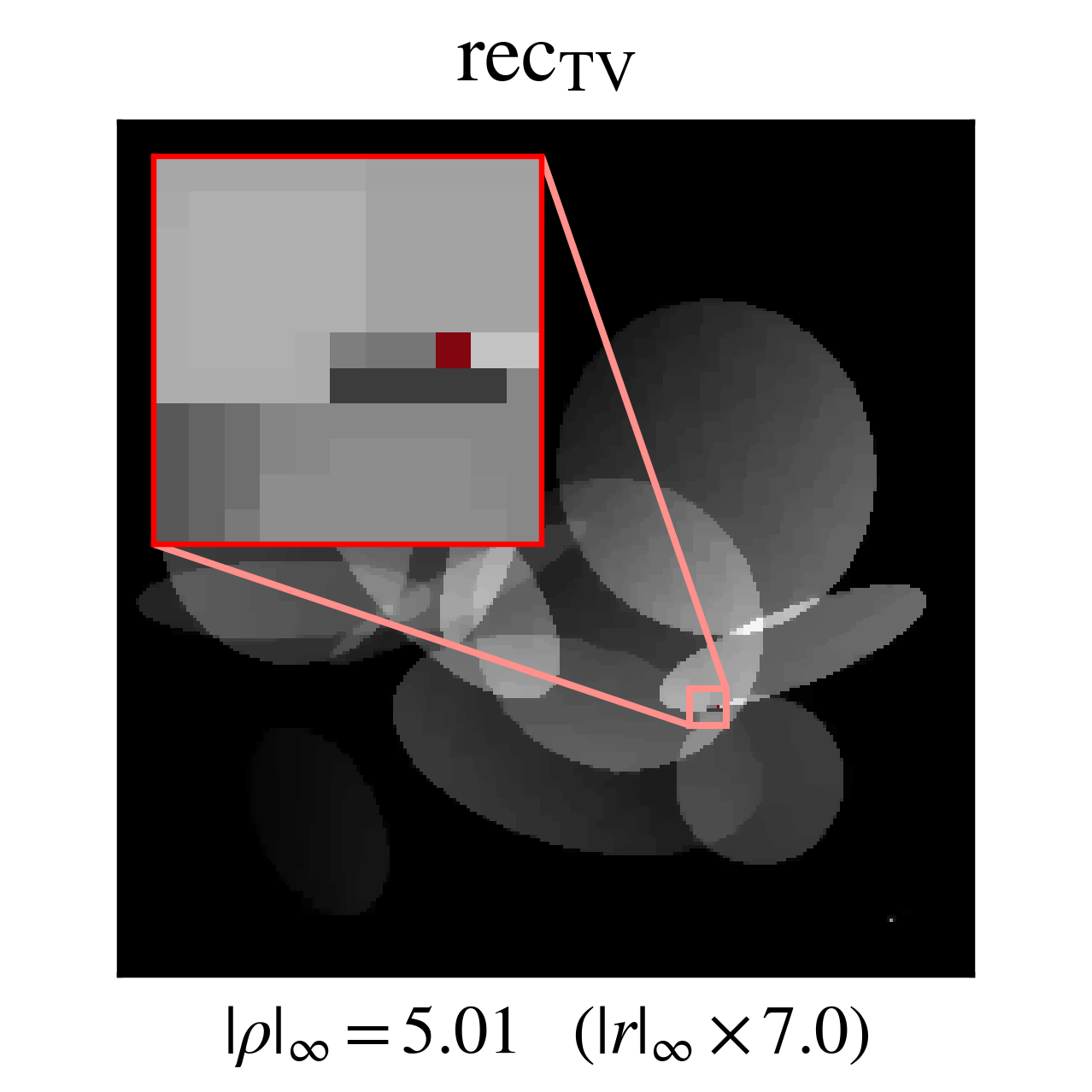}
\includegraphics[scale=0.58]{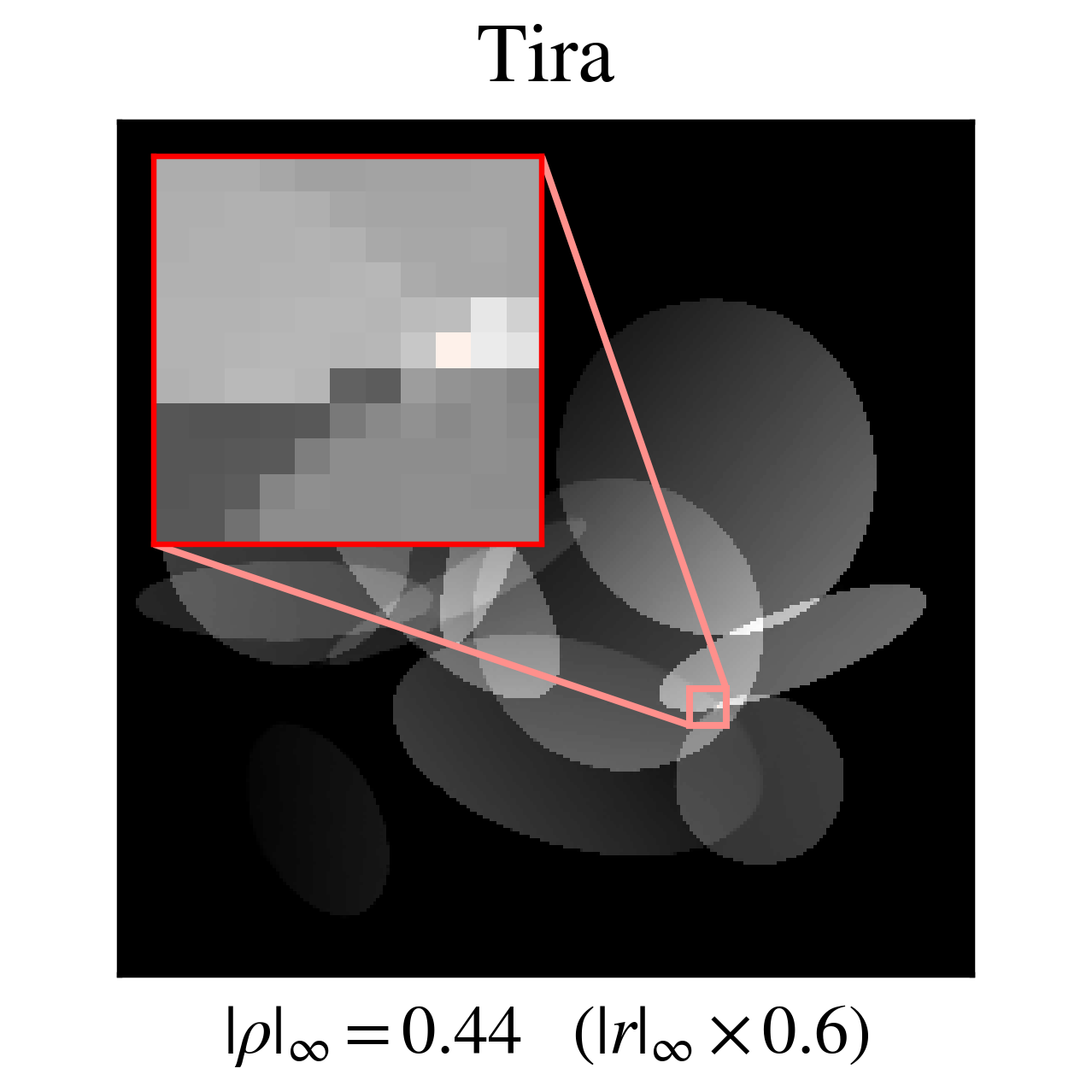}
\includegraphics[scale=0.58]{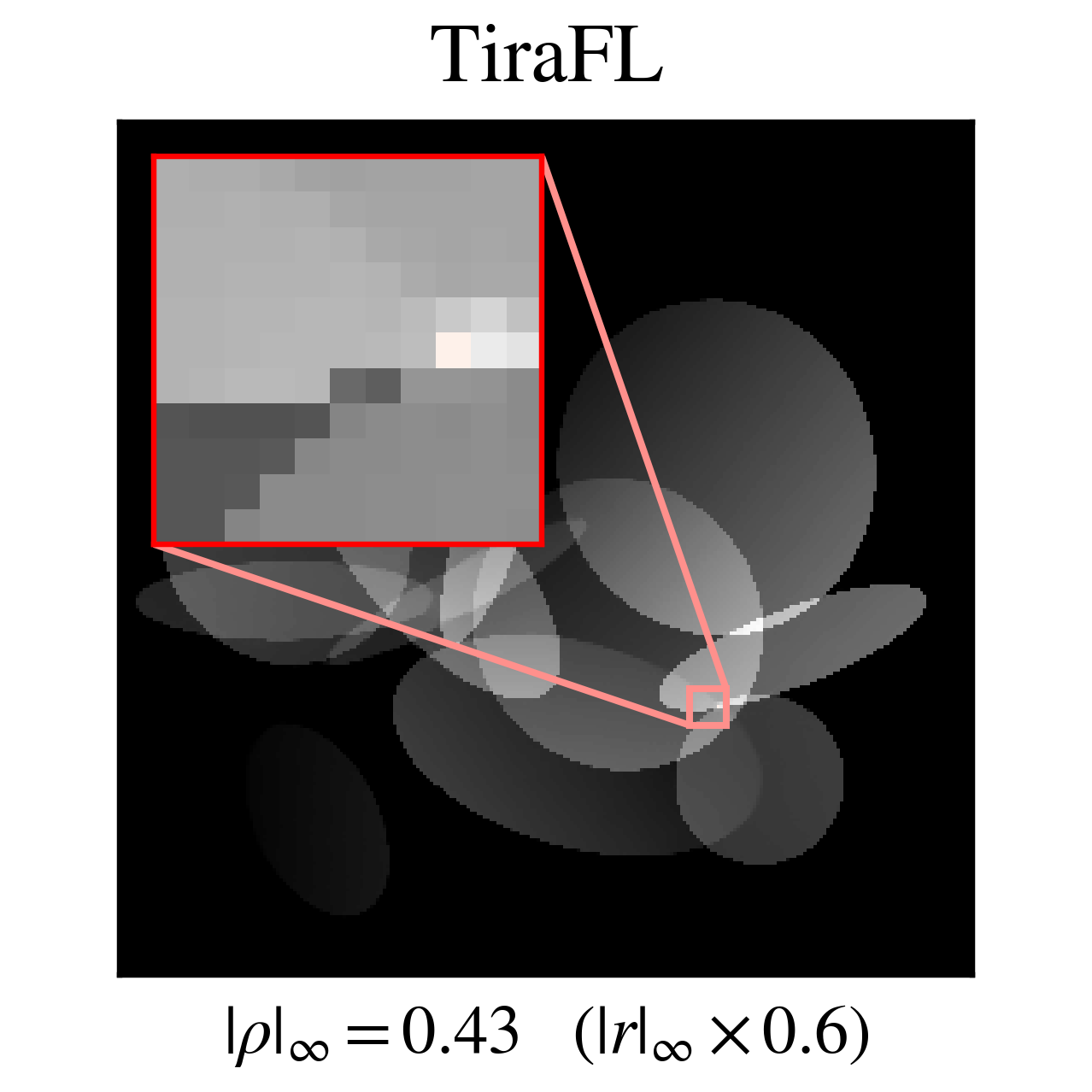}
\includegraphics[scale=0.58]{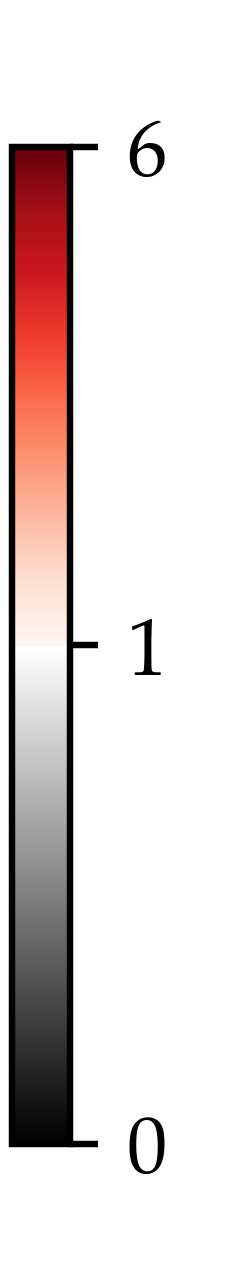}
\caption{Analogous to Figure \ref{fig:simple_transfer_tira_nr4p}, with the perturbation constructed with respect to \(\recon\).
The reconstruction artifact in the second image does not appear in the second and third images, and hence the perturbation is not transferable.
}
\label{fig:simple_transfer_tv_nr4p}
\end{figure}

\section{
Explaining localized artifacts}
\label{sec:maths}
Many results exist that guarantee exact recovery of sparse vectors from partial Fourier measurements (\ref{eq:noisymeasurements}), depending on properties of the underlying signal, the sampling mask \(\Omega\), and the reconstruction method.
The reconstruction artifacts seen in Figures \ref{fig:simple_signal_nr4p} and \ref{fig:simple_noise_lvls} consist mainly of a single pixel spike.
We now show how such exactness guarantees can imply the existence of low \(\ell^\infty\)-norm perturbations which give rise to precisely that type of artifact.

For simplicity, we consider the noiseless one-dimensional setting.
Let \(n\geq2\) be an even integer, and let \(\fourier_n\colon\C^n\to\C^n\) denote the one-dimensional discrete Fourier transform,
\begin{equation*}
    \fourier_n(z) =\left( \frac{1}{\sqrt{n}}\sum\limits_{j=0}^{n-1} z_j \expon{-2\pi\imag kj/n} \right)_{k=-n/2+1}^{n/2}
.\end{equation*}
Let \(\mask_n\subseteq\{-n/2+1, \ldots, n/2\}\) be a set of indices, and define the operator \(\mriop_n = \maskopn\circ\fourier_n\),
where \(\maskopn\) is the projection onto the index set \(\mask_n\).
Consider a signal \(x_n\in\C^n\).
We wish to perturb the measurement vector \(y_n = \mriop_n (x_n)\) such that a spike appears in the reconstruction.
Without loss of generality, let \(\delta_n=(1,0,\ldots,0)\in\C^n\) be our desired single spike artifact,
and simply define the perturbation
\begin{equation*}
    e_n = \mriop_n(\delta_n) = \left(1/\sqrt{n}\right)_{k\in\mask_n}
\end{equation*}
(which, incidentally, is an \(\ell^\infty\)-minimal perturbation given an \(\ell^2\)-budget).
The corresponding ``image perturbation'' is
\begin{equation*}
    r_n
    = \mriopinv_n (e_n)
    = \left( \frac{1}{n}
    \sum\limits_{k\in\mask_n}\expon{2\pi\imag kj/n}\right)_{j=0}^{n-1}
.\end{equation*}
If we can guarantee the exact recovery of \(x_n\) from \(\mriop_n x_n\) and of \(x_n+\delta_n\) from its measurements (see Theorems~\ref{thm:l1noiseless} and \ref{thm:tvnoiseless} below), 
\begin{equation*}
\mriop_n(x_n+\delta_n)=y_n + e_n=\mriop_n(x_n+r_n)
,
\end{equation*}
then a reconstruction artifact is created with an amplification factor \(\alpha_n=\norm{\delta_n}_\infty/\norm{r_n}_\infty=1/\norm{r_n}_\infty\).
The \(\ell^\infty\)-norm of \(r_n\) can easily be bounded from above,
\begin{equation*}
    \norm{r_n}_\infty =
    \max\limits_{j=0,\ldots,n-1}\abs{ \frac{1}{n} \sum\limits_{k\in\mask_n} \expon{2\pi\imag kj/n} }
    \leq
    \max\limits_{j=0,\ldots,n-1} \frac{1}{n} \sum\limits_{k\in\mask_n}\abs{\expon{2\pi\imag kj/n} }
    = \frac{m_n}{n}
,\end{equation*}
where \(m_n=\cardinality{\mask_n}\), so that
\[
\alpha_n\ge \frac{n}{m_n}.
\]
In other words, the amplification factor is at least as big as the subsampling factor, which is in accordance with Table~\ref{tab:quant}.
If the size of the sampling mask \(\mask_n\) grows slowly with the dimension \(n\),
then the amplification factor \(\alpha_n\) can be made arbitrarily large simply by increasing the resolution.
In fact,
since
\(
    \norm{r_n}_2=\norm{e_n}_2
    = \sqrt{m_n/n}
 \)
 and 
 \(
\norm{\delta_n}_2=1
,\)
the same is true for amplification in the \(\ell^2\)-norm, although the growth is at a slower rate.

For demonstration purposes, we now cite a known recovery guarantee for \(\ell^1\)-minimization,
which can then be translated to results on TV-minimization by a simple argument.
The number of measurements required depends on the sparsity of the signal, which we shall assume does not vary with the dimension.
This is true if \(x_n\) is a finite spike train and thus suitable for \(\ell^1\)-minimization.
We employ TV-minimization if \(x_n\) is a piecewise constant function, in which case the sparsity of \(\nabla x_n\) is constant with respect to \(n\).
Moreover, if \(x_n\) is \(s\)-sparse, then \(x_n+\delta_n\) is \((s+1)\)-sparse.
Similarly, if \(\nabla x_n\) is \(s\)-sparse, then \(\nabla(x_n+\delta_n)\) is \((s+2)\)-sparse (where \(\nabla\) is understood as the periodic finite difference operator).
Hence, exactness results hold for \(x_n+\delta_n\) as well as \(x_n\).

\begin{theorem}[{\cite[Theorem 1.1]{candes2011probabilistic}
as stated in
\cite[Theorem 12.20]{foucart2013invitation}}]
\label{thm:l1noiseless}
Let \(x\in\C^n\) be \(s\)-sparse, let \(\varepsilon>0\), and suppose that the indices of \(\mask_n\) are chosen uniformly at random from \(\{-n/2+1, \ldots, n/2\}\).
Then there exists a constant \(C>0\) such that if
\(m_n \geq Cs\log(n)\log(\varepsilon^{-1})\),
then \(x\) is the unique solution to
\begin{equation*}
	\label{eq:cs_noiseless_l1}
	\min\limits_{z\in\C^n}
	\norm{z}_1\ 
	\text{subject to}\ 
	\mriop_n z =  \mriop_n x,
\end{equation*}
with probability at least \(1-\varepsilon\).
\end{theorem}

Keeping \(\varepsilon\) at a small but constant value and using the minimal number of measurements means that \(m_n\lesssim\log n\).
Applying Theorem \ref{thm:l1noiseless} to \(x=x_n+\delta_n\), we see that \(\ell^1\)-minimization combined with this uniformly random sampling scheme leads to amplification of
\begin{equation*}
    \alpha_n
    \gtrsim \frac{n}{\log n} \xrightarrow[n\to\infty]{}\infty,
\end{equation*}
with high probability.
In the same way, the corresponding TV-minimization result implies that TV-minimization for gradient sparse signals also leads to amplification of \(\alpha_n\gtrsim n/\log n\).
\begin{theorem}
    \label{thm:tvnoiseless}
    Let \(x\in\C^n\) be such that \(\nabla x\) is \(s\)-sparse, let \(\varepsilon>0\),
    and suppose that \(\mask_n=\{0\}\cup\mask'\) where \(\mask'\) contains \(m_n\) indices chosen uniformly at random from \(\{-n/2+1, \ldots, n/2\}\).
    Then there exists a constant \(C>0\) such that if
    \(m_n \geq Cs\log(n)\log(\varepsilon^{-1})\),
    then \(x\) is the unique solution to
    \begin{equation*}
    	\label{eq:cs_noiseless_tv}
    	\min\limits_{z\in\C^n}
    	\norm{\nabla z}_1\ 
    	\text{subject to}\ 
    	\mriop_n z =  \mriop_n x,
    \end{equation*}
    with probability at least \(1-\varepsilon\).
\end{theorem}
\begin{proof}
    (The following argument can be used to translate results on recovery from Fourier measurements by \(\ell^1\)-minimization to TV-minimization and appears, for example, in \cite{candes2006robust}.)
    First, note that for any \(z\in\C^n\) and \(k=-n/2+1,\ldots,n/2\), we have
    \begin{align*}
        \left( \fourier_n \left( \nabla z\right)\right)_k
        & = \frac{1}{\sqrt{n}}\sum\limits_{j=0}^{n-1} \left( z_j - z_{(j-1) \bmod n} \right) \expon{-2\pi\imag kj/n}\\
        & = \frac{1}{\sqrt{n}}\sum\limits_{j=0}^{n-1} \left( z_j \right) \expon{-2\pi\imag kj/n}
        - \frac{1}{\sqrt{n}}\sum\limits_{j=0}^{n-1} \left( z_j \right) \expon{-2\pi\imag k(j+1)/n}\\
        & = \left( 1 - \expon{-2\pi\imag k/n} \right) \left( \fourier_n z\right)_k,
    \end{align*}
    and therefore
    \begin{equation*}
        \mriop_n z = \mriop_n x
        \quad
        \Longleftrightarrow
        \quad
        \mriop_n\left(\nabla z\right)
        = \mriop_n\left(\nabla x\right)
        \ \text{and}\ \mathbf{1}\cdot z = \mathbf{1}\cdot x
    \end{equation*}
    (where \(\mathbf{1}=(1,\ldots,1)\) and the condition \(\mathbf{1}\cdot z = \mathbf{1}\cdot x\) comes from the assumption that \(0\in\mask_n\)).
    
    By successively loosening the constraints, we see that
    \begin{align*}
        & 
    	\min\limits_{z\in\C^n}
    	\norm{\nabla z}_1\
        	\text{subject to}\ 
        	\mriop_n  z  =  \mriop_n  x\\
        =&
    	\min\limits_{z\in\C^n}
    	\norm{\nabla z}_1\
        	\text{subject to}\ 
        	\mriop_n \left( \nabla z \right) =  \mriop_n \left( \nabla x\right)
        \ \text{and}\ \mathbf{1}\cdot z = \mathbf{1}\cdot x\\
    	\geq &
    	\min\limits_{z\in\C^n}
    	\norm{\nabla z}_1\
        	\text{subject to}\ 
        	\mriop_n \left( \nabla z \right) =  \mriop_n \left( \nabla x\right)
        	\\
        \geq &
    	\min\limits_{z\in\C^n}
    	\norm{z}_1\ 
        	\text{subject to}\ 
        	\mriop_n z =  \mriop_n \left( \nabla x\right).
    \end{align*}
    By Theorem \ref{thm:l1noiseless}, \(\nabla x\) is the unique solution to the last of these minimization problems, and therefore \(x\) also solves the first one.
    Uniqueness follows from the fact that the only \(z\) which satisfies both \(\nabla z=\nabla x\) and \(\mathbf{1}\cdot z = \mathbf{1}\cdot x\) is \(z=x\).
\end{proof}

Of course, other sampling schemes exist, and if additional information on the signal is utilized, this result can be improved upon.
For example, if the minimal distance between two nonzero entries of \(x_n+\delta_n\) is greater than \(2n/c\) for some positive integer \(c\),
then \(\ell^1\)-minimization recovers \(x_n+\delta_n\) exactly as long as
the frequencies \(\{-c,c+1,\ldots,c\}\) are included in the sampling mask \(\mask_n\) \cite{candes2014towards}.
Likewise, if \(x_n+\delta_n\) is nonnegative and \(s\)-sparse, then it suffices that \(\{-s,s+1,\ldots,s\}\subseteq\mask_n\) \cite{de2012exact}.
In these cases \(m_n\) is constant, and therefore \(\alpha_n \gtrsim n\).
However, we cannot apply these specialized results to our example with TV-minimization, since \(\nabla(x_n+\delta_n)\) is neither nonnegative nor well separated.
It is worth mentioning the main result of \cite{poon2015role} on TV-minimization, which uses a random sampling mask that concentrates on the lower frequencies with \(m_n \gtrsim \log(n)\).

The single pixel reconstruction artifact may not be considered a meaningful feature and can be disregarded by a practitioner.
However, the artifacts seen in the experiments are more diverse and act as a proof of concept for localized adversarial artifacts for TV-regularization.
Other criteria for artifact severity may lead to different results.
The spike \(\delta_n\) is only an idealization of the observed artifacts which maximizes the \(\ell^\infty\)-norm within a fixed \(\ell^2\)-budget for the perturbation.
A similar argument can be made for the existence of different types of sparse artifacts but with a lower \(\ell^\infty\)-amplification factor.
For any artifact \(\rho_n\in\C^n\), the corresponding image perturbation can be bounded by the triangle inequality
\begin{equation*}
    \norm{r_n}_\infty=
    \norm{\mriopinv_n \mriop_n(\rho_n)}_\infty \leq \norm{\rho_n}_1\frac{m_n}{n}
    ,
\end{equation*}
so that an amplification of
\begin{equation*}
    \alpha_n
    \geq
    \frac{
    \norm{\rho_n}_\infty
    }{
    \norm{\rho_n}_1
    }
    \frac{n}{m_n}
\end{equation*}
can be attained.
Depending on the sparsity properties of \(\rho_n\), exact recovery of \(x_n+\rho_n\) can be guaranteed by the theorems above.
For example, if \(\rho_n=(1,\ldots,1,0,\ldots,0)\) is a discrete rectangular function, then \(\nabla\rho_n\) is 2-sparse and Theorem \ref{thm:tvnoiseless} can be applied.

\section*{Acknowledgments}
This material is based upon work supported by the Air Force Office of Scientific Research under award number FA8655-20-1-7027. GSA is a member  of the ``Gruppo Nazionale per l'Analisi Matematica, la Probabilità e le loro Applicazioni'', of the ``Istituto Nazionale di Alta Matematica''.

\bibliography{ref.bib}
\end{document}